\documentclass[11pt]{amsart}

\usepackage{amssymb}
\usepackage{amsthm}
\usepackage{amsfonts}
\usepackage{amsmath}
\usepackage{pmboxdraw}
\usepackage{verbatim} 
\usepackage{graphicx}
\usepackage{color}
\usepackage[colorlinks=true, citecolor=blue, filecolor=black, linkcolor=black, urlcolor=black]{hyperref}
\usepackage{cite}
\usepackage[normalem]{ulem}
\usepackage{subcaption}
\usepackage{bbm}
\usepackage{bm}
\usepackage{mathtools}
\usepackage{todonotes}
\usepackage{enumitem}
\usepackage{multirow,makecell}  
\usepackage{tablefootnote}
\usepackage{braket}

\def\cE{\mathcal{E}}
\def\cF{\mathcal{F}}
\def\cL{\mathcal{L}}
\def\cO{\mathcal{O}}

\def\cC{\mathcal{C}}
\def\cK{\mathcal{K}}

\def\cS{\mathcal{S}}

\def\cH{\mathcal{H}}

\def\Z{\mathbb{Z}}
\def\cJ{\mathcal{J}}

\def\C{\mathbb{C}}

\def\E{\mathbf{E}}
\def\N{\mathbb{N}}

\def\R{\mathbb{R}}

\newcommand{\diag}{{\textup{diag}}}

\newcommand{\im}{\operatorname{Im}}

\def\zbar{\overline{z}}
\def\wbar{\overline{w}}

\theoremstyle{plain}
\newtheorem*{thm*}{Theorem}
\newtheorem{thm}{Theorem}[section]
\newtheorem{lem}[thm]{Lemma}
\newtheorem{cor}[thm]{Corollary}
\newtheorem{prop}[thm]{Proposition}
\newtheorem*{prop*}{Proposition}
\newtheorem*{lem*}{Lemma}
\newtheorem{rem}[thm]{Remark}

\theoremstyle{definition}
\newtheorem*{eg*}{Example}
\newtheorem*{egs*}{Examples}
\newtheorem*{Q*}{Question}

\theoremstyle{remark}
\newtheorem*{rmk*}{Remark}
\newtheorem*{rmks*}{Remarks}

\numberwithin{equation}{section}

\begin{document}
\title[Determinantal structure of the overlap for IGinUE]{Determinantal structure of the conditional expectation of \\
the overlaps for the induced Ginibre unitary ensemble
}

\author{Kohei Noda}
\address{Joint Graduate School of Mathematics for Innovation, Kyushu University, West Zone 1, 744 Motooka, Nishi-ku, Fukuoka 819-0395, Japan}
\email{noda.kohei.721@s.kyushu-u.ac.jp}



\date{\today}

\begin{abstract}
As is widely known, a non-Hermitian matrix exhibits distinct left and right eigenvectors, which form a bi-orthogonal system. 
Chalker and Mehling initiated the study of the joint statistics of the eigenvalues and the overlaps defined by the left and right eigenvectors of the Ginibre unitary ensemble. 
Later, Akemann et al. continued their investigation by studying the $k$-th correlation function weighted by the on- and off-overlaps of the Ginibre unitary ensemble. \par
In this paper, as a natural extension of their work, we investigate the $k$-th correlation function weighted by the on- and off-diagonal overlaps of the induced Ginibre unitary ensemble. 
Similar to the Ginibre unitary ensemble case, we will demonstrate the determinantal structure. 
As a result, we will confirm the universality of the $k$-th correlation function weighted by the on- and off-diagonal overlaps in both the bulk and edge scaling limits in the strongly non-unitary regime. 
Furthermore, in the weakly non-unitary regime and at the singular origin, 
we will report new relationships between the overlap and such spectral regimes. 
\end{abstract}

\maketitle
\section{Introduction}
Since the spectral analysis of the eigenvalues in random matrix theory was initiated by Wigner, 
the universalities for the eigenvalues of random matrices have been studied from the both macroscopic and microscopic perspectives so far. 
Moreover, the statistics of the point processes arising from the scaling limits of their eigenvalues have been currently received much attention. 
Indeed, they can be applied to a broad range of fields, for instance, biology \cite{ABC}, quantum optics \cite{BFD}, and machine learning \cite{KT}.  
The eigenvectors statistics of the random matrices belonging to symmetric classes such as symmetric/hermitian/symplectic matrices have been also studied, and it is known that they also fall into the universality classes, see \cite{RVW} and references therein. 
On the other hand, the eigenvectors statistics of the non-Hermitian random matrices have received less attention compared to the local statistics of the eigenvectors of random matrices belonging to the symmetric classes. 
In 1998, Chalker and Mehling \cite{CM_1998,CM_2000} introduced the overlap defined by left and right eigenvectors of non-Hermitian random matrices.
In particular, they studied the overlap for the Ginibre unitary ensemble, 
which is $N\times N$ random matrices with independent, identically distributed to standard complex Gaussian random variables elements. 
For the physical motivations of the overlap of the non-Hermitian random matrices, we refer to \cite{FM,Fyodorov_2003,JNNPZ,MS}. 
Mathematically, it plays an important role to analyze the stability of the eigenvalues of non-Hermitian matrices and describe the stochastic dynamics of the non-Hermitian matrix-valued Brownian motions \cite{Bourgade_Dubach:2021,Esaki,Grela_2018,Yabuoku}.
Here, we mention two breakthrough papers about the correlations of the overlaps for the Ginibre unitary ensemble. 
Bourgade and Dubach \cite{Bourgade_Dubach:2021} studied the distribution of the on-diagonal overlap of the Ginibre unitary ensemble and the correlations of the on- and off-diagonal overlaps at microscopic and mesoscopic scales. 
Their method is probabilistic and depends on the moment bound of the characteristic polynomials of the Ginibre unitary ensemble \cite{WW}.   
On the other hand, Akemann, Tribe, Tsareas, and Zaboronski \cite{ATTZ} studied the determinantal structure of the $k$-th correlation function associated with the Gaussian potential deformed by on-and off-diagonal overlaps, 
and they showed the new local statistics at bulk and edge scaling limits. 
As a consequence, they showed the scaling limits of the means of the on- and off-diagonal overlaps at both regimes. 
Their method is based on the planar orthogonal polynomials and the uniform asymptotic expansion of the incomplete Gamma functions.
The main difficulty is to construct the planar orthogonal polynomials associated with the Gaussian weight function deformed by the overlaps and to simplify its correlation kernel. 
Our work and approach in this paper are motivated by the later work, and we generalize their results for the Ginibre unitary ensemble into the induced Ginibre unitary ensemble, which has the spectral singularity at the origin and the weakly non-unitary regime (see Section \ref{s2} for the precise definitions.) 
Also, let us mention the conditional expectation of the on-diagonal overlap of the Ginibre orthogonal and the elliptic Ginibre orthogonal ensemble. 
Fyodorov \cite{Fyodorov_2018} studied the conditional expectation of the on-diagonal overlap of the real eigenvalue for the Ginibre orthogonal ensemble via the super-symmetric method. 
His result reveals the full-probabilistic density of the on-diagonal overlap, 
and he found that the tail decay of the on-diagonal overlap of the real eigenvalue is the inverse quadratic behavior,
which is different from the case of Ginibre unitary ensemble case (indeed, the tail behavior of the conditional expectation of the on-diagonal overlap is the inverse cubic).
Later, Fyodorov and Tarnowski \cite{Fyodorov_2021} studied the same quantity for the elliptic Ginibre orthogonal ensemble,
and they showed the scaling limit of the on-diagonal overlap of the real eigenvalue in the strongly non-unitary and weakly non-unitary regimes via the same method. 
Recently, W\"{u}rfel, Crumpton, and Fyodorov studied the conditional expectation of the on-diagonal overlap for any complex or real eigenvalues of the Ginibre orthogonal ensemble via the incomplete Schur decomposition and a super-symmetric method. 
Finally, for the works of the overlap related to the integrable non-Hermitian random matrices such as the truncated unitary ensemble, spherical unitary ensemble, and Ginibre symplectic ensemble, we only refer to \cite{Akemann_Foster_Kieburg:2020,BZ,BSV:2017,CR2022,Dubach_2021w1,Dubach_2021w2,Dubach_2023,Nowak_2018}.
Apart from the integrable non-Hermitian random matrices, 
\cite{Cipolloni22,Cipolloni23c,ErdJi} studied the overlap of the general non-Hermitian random matrices satisfying appropriate assumptions, and they showed that the overlap for the wide classes of the non-Hermitian random matrices also fall into the universality class. 
As we already announced, 
the aim in this paper is to study the $k$-th correlation function associated with the overlap weight function for the induced Ginibre unitary ensemble, 
which can be regarded as the generalization of the results in \cite{ATTZ}. 
In particular, we confirm that the overlap of the induced Ginibre unitary ensemble fall into the universality classes. 
Also, we will study the asymptotic behavior of the on- and off-diagonal overlaps of that model in the singular origin and weakly non-unitary regime, and we report the new relationships between the overlap and spectral regimes such as the weakly non-unitary regime and the singular origin. 
\subsubsection*{The remainder organization in this paper} 
The rest of the present paper is organized as follows: 
in section~\ref{s2}, we will introduce the main objects in this paper, induced Ginibre unitary ensemble and its overlap between left and right eigenvectors. And we collect some facts for the local statistics of the induced Ginibre unitary ensembles. 
In section~\ref{s3}, we will present our main results.
In section~\ref{s4}, we establish the finite $N$-kernel based on the moment method, and we will simplify the its finite $N$-kernel. 
In section \ref{s5}, we complete the proof of our main results. 
\section{Preliminaries}\label{s2}
In this section, we introduce the induced Ginibre unitary ensemble, 
and we collect some facts of the local statistics for the eigenvalues of the induced Ginibre unitary ensemble. 
\subsubsection{Induced Ginibre unitary ensemble}
First, we recall the way to construct the induced Ginibre unitary ensemble: 
\begin{prop}[\cite{Fischmann},\cite{Sungsoo2_2022}] \label{ConsPro}
For $n,N\in\N$ with $n\geq N$, let $g:\C^{N\times N}\to\R_{\geq 0}$ and $\mathbf{G}$ be a bi-unitary invariant rectangular $n\times N$ random matrix with joint matrix element distribution proportional to $g\left(\mathbf{G}^\dagger\mathbf{G}\right)$. 
Then, the joint probability distribution of the matrix $\mathbf{A}:=\left(\mathbf{G}^\dagger \mathbf{G}\right)^{\frac{1}{2}}\mathbf{U}$ 
with $\mathbf{U}$ a Haar unitary matrix with size $N$ is proportional to 
\begin{equation}
\left(\det\mathbf{A}^\dagger\mathbf{A}\right)^{n-N}g\left(\mathbf{A}^\dagger\mathbf{A}\right).
\label{jpdfMatrix}
\end{equation}
\end{prop}
In this proposition \ref{ConsPro}, we set $g(X)=e^{-\mathrm{Tr}(X)}$ for $X\in\C^{N\times N}$,
and we fix $\mathbf{G}$ as $n\times N$ rectangular complex Ginibre ensemble,
whose elements are i.i.d. standard complex Gaussian random variables.
Then, we can construct the random matrix $\mathbf{A}$ with joint probability distribution function \eqref{jpdfMatrix}. 
We call such random matrix $\mathbf{A}$ the {\it induced Ginibre unitary ensemble}.
By change of the variables from the matrix elements of $\mathbf{A}$ to the eigenvalues $\boldsymbol{z}_{(N)}=(z_1,...,z_N)\in\C^N$ of $\mathbf{A}$,  
we can find that its the joint probability distribution function of the eigenvalues of IGinUE $\mathbf{A}$ is given by
\begin{equation}
\label{IGinUE}
\frac{1}{Z_N}\prod_{1\leq j<k\leq N}|z_j-z_k|^2\prod_{k=1}^N|z_k|^{2\alpha}e^{-|z_k|^2},\quad \alpha=n-N,
\end{equation}
where the partition function $Z_N$ is given by
\begin{equation}
Z_N=N!\prod_{k=0}^{N-1}\Gamma(k+\alpha+1).
\label{eq:PFOE}
\end{equation}
For the detailed discussions, we refer to \cite{Fischmann} and the monograph \cite[section 2.4.]{Sungsoo2_2022}.
Rescaling $z_k\mapsto \sigma z_k$ for $k=1,2,...,N$ with $\sigma>0$, 
we can also write 
\begin{equation}
\frac{1}{\widetilde{Z}_N}\prod_{1\leq j<k\leq N}|z_j-z_k|^2\prod_{k=1}^Ne^{-\sigma^2Q(z_k)},\quad 
Q(z):=|z|^2-\frac{2\alpha}{\sigma^2}\log|z|, \label{eq:jpdf1}
\end{equation}
where $\widetilde{Z}_N=N!\prod_{k=0}^{N-1}\frac{\Gamma(k+\alpha+1)}{\sigma^{2k+2\alpha+2}}$. 
Here, $\alpha$ is the integer-valued parameter when we consider the random matrix model $\mathbf{A}$, 
but we can still consider the general $\alpha>-1$. 
Then, it is known that the $k$-th correlation function of the eigenvalues forms the determinantal point process with the correlation kernel given by 
\[
\mathbf{R}_{N,k}(\boldsymbol{z}_{(k)})
=\det\bigl(
\mathbf{K}_N(z_i,z_j)
\bigr)_{i,j=1}^{k},\quad
\mathbf{K}_N(z,w)
=
\sigma^{2}\sum_{j=0}^{N-1}\frac{(\sigma^2z\overline{w})^{j+\alpha}}{\Gamma(j+\alpha+1)}e^{-\frac{\sigma^2}{2}(|z|^2+|w|^2)}. 
\]
\subsubsection{Macroscopic and microscopic properties of the spectrums of IGinUE}
We briefly discuss the macroscopic properties of the spectral droplet of IGinUE for the parameters with $\sigma^2=Na_N$ and $\alpha=b_N>-1$ and the local statistics of the eigenvalues point processes in each regime following \cite{Sungsoo6_2022, Sungsoo1_2022}. 
From the facts of the logarithmic potential theory \cite{SFT,Sungsoo6_2022}, 
the spectral droplet associated with the potential $Q$ in \eqref{eq:jpdf1} tends to 
\[
S=\{z\in\C:r_1\leq |z|\leq r_2\}\quad\text{with $r_1=\sqrt{\frac{b_N}{Na_N}}$ and $r_2=\sqrt{\frac{b_N+N}{Na_N}}$}.
\]
Then, the spectral droplet is clarified by the following three regimes depending on parameters $a_N,b_N$. 
\begin{description}
\item[(a) Strongly non-unitary regime]\label{regime1}
Without loss of generality, we set $a_N=1$ and $b_N=Nb$ for $b>0$. Then, the spectral droplet tends to 
\begin{equation}
S_{\mathrm{reg}}:=\left\{z\in\C:\sqrt{b}\leq|z|\leq \sqrt{b+1}\right\}  \label{eq:regular}
\end{equation}
as $N\to\infty$.
Hence, the spectral droplet becomes an annulus. 
In this case, the local statistics of the point process are fall into the universality classes. 
Indeed, in \cite{Fischmann}, they showed that the limiting reproducing kernel at bulk points of $S_{\mathrm{reg}}$ is featured by $K_{\mathrm{bulk}}(z,w)=G(z,w):=e^{z\overline{w}-\frac{1}{2}|z|^2-\frac{1}{2}|w|^2}$. 
On the other hand, the limiting kernel at edge points, that is, at inner boundary points and outer boundary points of $S_{\mathrm{reg}}$ is featured by 
\[
K_{\mathrm{edge}}(z,w)=G(z,w)F(z+\overline{w}),
\]
where $F(x)$ is the complementary error function defined by
\begin{equation}
F(x):=\frac{1}{\sqrt{2\pi}}\int_{x}^\infty e^{-\frac{s^2}{2}}ds=\frac{1}{2}\mathrm{erfc}\left(\frac{x}{\sqrt{2}}\right)
\quad\text{for $x\in\C$}.\label{errorF}
\end{equation}
These reproducing kernels are same as the ones of the limiting kernels at the bulk and edge regimes for the Ginibre unitary ensemble. 
\item[(b) Weakly non-unitary regime (almost circular regime)]\label{regime2}
In this regime, we set the parameters as
\begin{equation}
a_N=\frac{N}{\rho^2},\quad b_N=N\left(\frac{N}{\rho^2}-\frac{1}{2}\right). \label{eq:ATP}
\end{equation}
Then, the spectral droplet tends to 
\begin{equation}
S_{\mathrm{weak}}:=\left\{z\in\C:r_1\leq|z|\leq r_2\right\}\quad
\text{with $r_1=1-\frac{\rho^2}{4N}+o\left(N^{-1}\right),\quad
r_2=1+\frac{\rho^2}{4N}+o\left(N^{-1}\right)$}, \label{eq:almost}
\end{equation}
as $N\to\infty$.
Hence,  the spectral droplet becomes almost unit circle, 
and hence, the spectral distribution globally looks like the circular unitary ensemble. 
Due to the above macroscopic spectral property, the local statistics of IGinUE can be formally expected to interpolate the local statistics between the Ginibre unitary ensemble and circular unitary ensemble,
which is the random matrix uniformly distributed to the Haar measure on the unitary group. 
Indeed, in \cite{Sungsoo6_2022,Ameur_2021}, it was shown that the limiting reproducing kernel in the weakly non-unitary regime is featured by
\[
K_{\rho}(z,w)=G(z,w)L_{\rho}(z+\overline{w}),
\]
where
\begin{equation}
L_{\rho}(z)=\frac{1}{\sqrt{2\pi}}\int_{-\frac{\rho}{2}}^{\frac{\rho}{2}}e^{-\frac{1}{2}(z-\xi)^2}d\xi\quad\text{for $z\in\C$ and $\rho>0$}.
\label{banderror}
\end{equation}
Then, if we take a limit $\rho\to\infty$, then $K_{\rho}(z,w)\to K_{\mathrm{bulk}}(z,w)$ in the distributional sense of point process. On the other hand, rescaling by $a=\frac{\rho}{2}$, $\frac{1}{a^2}K_{\rho}\left(\frac{z}{a},\frac{w}{a}\right)\to K_{\mathrm{sin}}(z,w)=\frac{\sin(x-y)}{x-y}$ for $x=\im(z),y=\im(w)$ in the distributional sense of point process. 
In this sense, from both macroscopic and microscopic view points, IGinUE can be regarded as the random matrix interpolating role between the Ginibre unitary ensemble (non-normal matrix) and the circular unitary ensemble (normal matrix). 
\item[(c) At the singular origin]\label{regime3}
Without loss of generality, we ser $a_N=1$ and $b_N=b>0$. Then, the spectral droplet tends to 
\begin{equation}
S_{\mathrm{sing}}:=\left\{z\in\C:r_1\leq|z|\leq r_2\right\}\quad\text{with $r_1=\frac{b}{\sqrt{N}}$ and $r_2=1+O\left(\frac{1}{N}\right)$}, \label{eq:singular}
\end{equation}
as $N\to\infty$, which contains the singular origin. The limiting reproducing kernel of the point process at the singular origin is featured by the two-parametric Mittag-Leffler function $E_{a,b}(z)$ given by
\begin{equation}
E_{a,b}(z):=\sum_{k=0}^\infty\frac{z^k}{\Gamma(ak+b)},
\label{ML}
\end{equation}
with $a=1$ and $b>0$. 
For further discussion in this regime, we refer to \cite{Sungsoo1_2022, Sungsoo3_2022}
\end{description}
\begin{figure}[htbp]
  \begin{minipage}[b]{0.3\linewidth}
    \centering
    \includegraphics[keepaspectratio, scale=0.3]{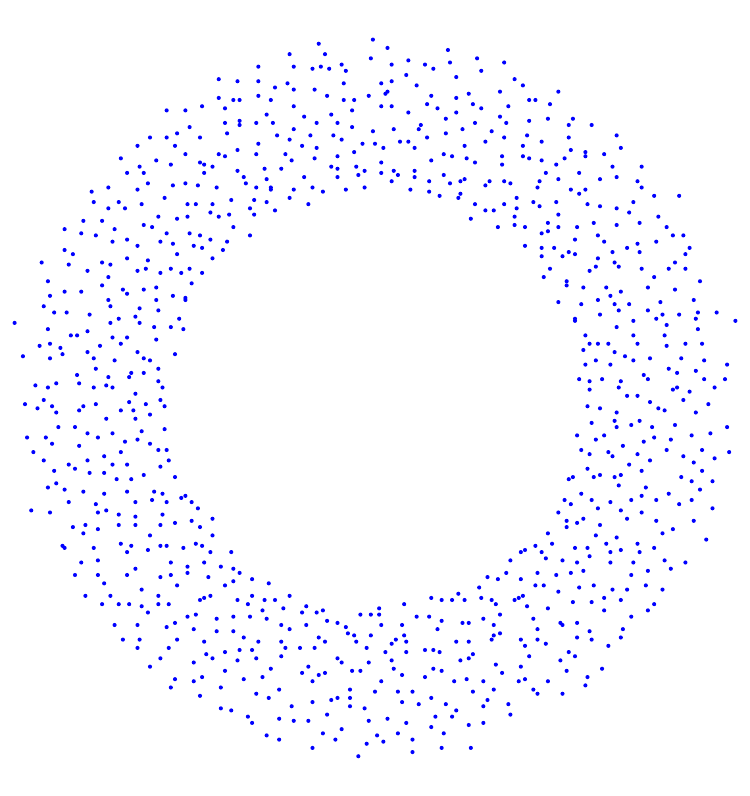}
    \subcaption{Strong non-unitary regime}
  \end{minipage}
  \begin{minipage}[b]{0.3\linewidth}
    \centering
    \includegraphics[keepaspectratio, scale=0.3]{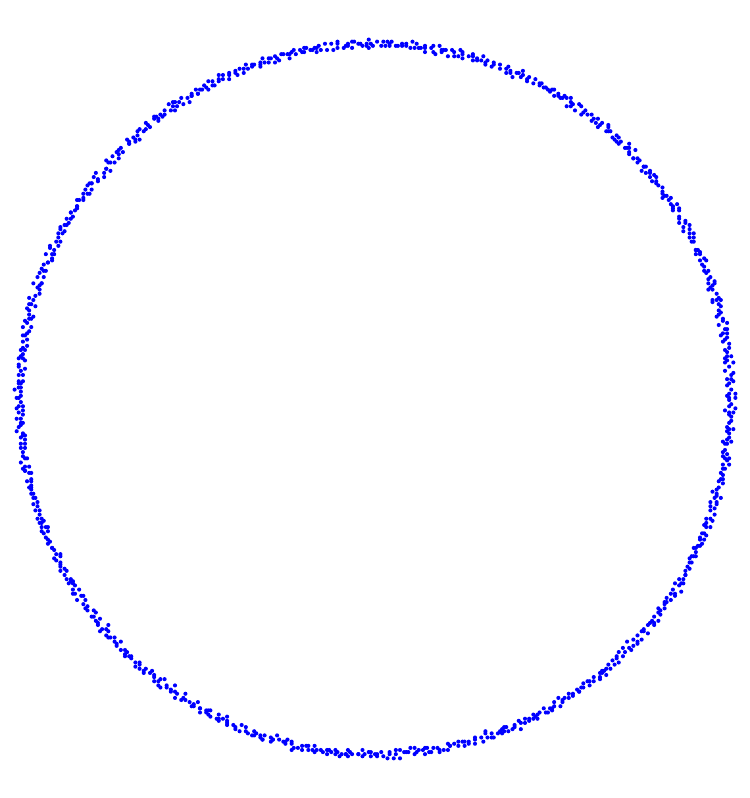}
    \subcaption{Weak non-unitary regime}
  \end{minipage}
    \begin{minipage}[b]{0.3\linewidth}
    \centering
    \includegraphics[keepaspectratio, scale=0.3]{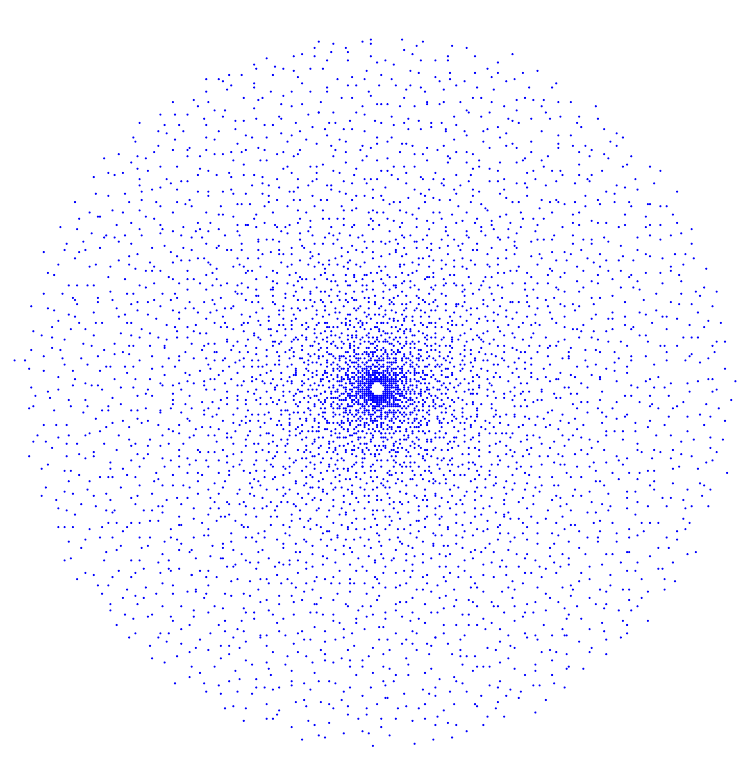}
    \subcaption{Singular regime}
  \end{minipage}
  \caption{These pictures are spectral droplets of the induced Ginibre ensemble depending on the parameters $a_N,b_N$.}
\end{figure}
Lastly, let us briefly mention the local statistics of the point processes for the other ensembles. 
Beyond IGinUE, in \cite{Sungsoo2_2022}, 
they studied the local statistics of point process of the random normal matrices with a radially symmetric external potential under the soft, hard/soft, hard edge, and boundary confinement in the weakly non-unitary regime. 
We also mention the recent developments of the local statistics of the point processes of two-dimensional Pfaffian Coulomb gases such as the Ginibre symplectic ensemble, the induced Ginibre symplectic ensemble, and the induced spherical symplectic ensemble. 
For the detailed discussions and the developments, we refer to \cite{Sungsoo1_2022,Sungsoo3_2022,Sungsoo4_2022}.
\subsection{Overlaps of IGinUE}
In order to describe our results, we explain the overlap of IGinUE in this subsection. 
First, for a given matrix $\mathbf{G}\in\mathbb{M}_N(\C)$ with simple spectrum $\{z_1,z_2,...,z_N\}$, 
we denote a left and right eigenvectors associated with the complex eigenvalue $z_j$ by $L_j,R_j$ defined by
\[
\mathbf{G}_NR_j=z_j R_j,\quad L_j^t\mathbf{G}_N=L_j^tz_j,
\]
chosen under the bi-orthogonal condition 
\begin{equation}
\braket{L_i|R_j}=\delta_{i,j}, \label{eq:biONS}
\end{equation}
then elements of matrix of overlaps $\cO$ is defined are 
\[
\cO_{i,j}=\braket{L_i|L_j}\braket{R_i|R_j}.
\]
Here, $\braket{\boldsymbol{a}|\boldsymbol{b}}=\sum_{j=1}^{d}\overline{a_i}b_i$ for $\boldsymbol{a}=(a_1,...,a_d),\boldsymbol{b}=(b_1,...,b_d)\in\C^{d}$ is the inner product on $\C^d$. 
From now on, we consider the induced Ginibre unitary ensemble constructed in Proposition \ref{ConsPro} as $\mathbf{G}_N$. 
Then, following the discussions of \cite{ATTZ,Akemann_Foster_Kieburg:2020,Bourgade_Dubach:2021,CM_1998,CM_2000}, 
we have that the following equality in law holds conditionally on $\{z_1,...,z_N\}$:
\begin{equation}
\cO_{1,1}\overset{d}{=}\prod_{j=2}^N\left(1+\frac{|X_{j}|^2}{|z_1-z_{j}|^2}\right), \label{O11}
\end{equation}
where $X_j$'s are independent complex Gaussian random variables with variance 1.
In particular, conditionally on $\{z_1,...,z_N\}$,
\begin{equation}
\E\left[\cO_{1,1}\right]=\prod_{j=2}^N\left(1+\frac{1}{|z_1-z_{j}|^2}\right). \label{AO11}
\end{equation}
Similarly, since we have that conditionally on $\{z_1,...,z_N\}$,
\begin{equation}
\cO_{1,2}\overset{d}{=}
-\frac{1}{|z_1-z_2|^2}\prod_{k=3}^N\left(1+\frac{|X_k|^2}{(z_1-z_k)(\overline{z_2-z_k})}\right),
 \label{O12}
\end{equation}
and we also have that conditionally on $\{z_1,...,z_N\}$,
\begin{equation}
\E\left[\cO_{1,2}\right]
=-\frac{1}{|z_1-z_2|^2}\prod_{j=3}^N\left(1+\frac{1}{(z_1-z_j)\overline{(z_2-z_j)}}\right).\label{AO12}
\end{equation}
Although we do not give the proofs of \eqref{O11}, \eqref{AO11}, \eqref{O12}, and \eqref{AO12} since their proofs are essentially same as the one of Ginibre unitary ensemble case. 
For the detailed discussions, we refer to \cite{Bourgade_Dubach:2021,CM_1998,CM_2000}.
\section{Main results}\label{s3}
\subsection{Determinantal structure of the $k$-th conditional expectation of the overlaps for the induced Ginibre unitary ensemble} 
Now, we state our main results. 
Firstly, we introduce the $k$-th conditional expectation of the overlap:
\begin{equation}
D_{p,q}^{(N,k)}(\boldsymbol{z}_{(k)})
:=
\left\langle\sum_{i_1\neq\cdots \neq i_k}\cO_{p,q}\delta(\lambda_{i_1}-z_{i_1})\cdots\delta(\lambda_{i_k}-z_{i_k})\right\rangle.\label{eq:d11}
\end{equation}
Here, $\langle \cdot \rangle$ denotes the expectation with respect to the induced Ginibre unitary ensemble \eqref{IGinUE}.
By the exchangeability of the eigenvalues, in this paper, we will mainly study $D_{11}^{(N,k)}(\boldsymbol{z}_{(k)})$. 
More precisely, 
\begin{equation}
D_{11}^{(N,k)}(\boldsymbol{z}_{(k)})=
\frac{N!}{(N-k)!}\frac{(z_1\overline{z_1})^{\alpha}e^{-|z_1|^2}}{Z_N}
\int_{\C^{N-k}}|\Delta_{N-1}(z_2,z_3,...,z_N)|^2
\prod_{j=2}^N\omega(z_j,\zbar_j|z_1,\zbar_1)\prod_{j=k+1}^NdA(z_j),\label{eq:D11}
\end{equation}
where the weight function $\omega$ is defined as 
\begin{equation}
\omega(z,w|u,v)=(1+(z-u)(w-v))(zw)^{\alpha}e^{-zw}\quad\text{for $z,w,u,v\in\C$}.
\label{eq:weight}
\end{equation}
Here, $Z_N=N!\prod_{k=0}^{N-1}\Gamma(k+\alpha+1)$.
We will also study the $k$-th conditional expectation of the off-diagonal overlaps, but 
we mainly consider the case of the on-diagonal overlap due to the following decoupling Lemma~\ref{thm:lem1}, 
which tells us that the $k$-th conditional expectation of the off-diagonal overlap can be written in terms of the $k$-th conditional expectation of the on-diagonal overlap up to some function. 
We identify $(z_1,\overline{z_1},z_2,\overline{z_2},z_3,\overline{z_3},...,z_k,\overline{z_k})$ as independent variables on $\C^{2k}$. 
\begin{lem}\label{thm:lem1}
Let 
$\widehat{T}:\C^{2k}\to\C^{2k}$ be a map given by 
\begin{equation*}
\widehat{T}f(z_1,\overline{z_1},z_2,\overline{z_2},z_3,\overline{z_3},...,z_k,\overline{z_k})=f(z_1,\overline{z_2},z_2,\overline{z_1},z_3,\overline{z_3},...,z_k,\overline{z_k})
\end{equation*}
for any function $f$ on $\C^{2k}$. Then, we have 
\begin{equation}
D_{12}^{(N,k)}(\boldsymbol{z}_{(k)})=-\frac{e^{-|z_1-z_2|^2}}{1-|z_1-z_2|^2}\widehat{T}D_{11}^{(N,k)}(\boldsymbol{z}_{(k)}).\label{eq:d12}
\end{equation}
\end{lem}
\begin{rem}
Lemma \ref{thm:lem1} was originally found for the case of the Ginibre unitary ensemble in \cite{ATTZ}. 
Since we can prove Lemma~\ref{thm:lem1} similar to \cite{ATTZ}, we omit the proof here. 
Also, notice that this decoupling lemma can also work well for the other ensembles such as the induced spherical unitary ensemble, truncated unitary ensemble, and the elliptic Ginibre unitary ensemble. 
\end{rem}
If we regard complex variables $z$ and $\zbar$ as independent variables, 
there exists a family of bi-orthogonal holomorphic polynomials 
$\{P_k\left(\cdot,|\lambda,\overline{\lambda}\right),Q_k\left(\cdot,|\lambda,\overline{\lambda}\right)\}_{k\in\N_{\geq0}}$ on $\C$ with respect to the inner product defined as
\begin{align}
\label{innerP}
\langle P_i(\cdot|\lambda,\overline{\lambda}),Q_j(\cdot|\lambda,\overline{\lambda})\rangle_{\omega}
:=&\int_\C 
\overline{P_k(z|\lambda,\overline{\lambda})}Q_k(z|\lambda,\overline{\lambda})
\omega(z,\overline{z}|\lambda,\overline{\lambda})dA(z) \\
=&\langle P_i(\cdot|\lambda,\overline{\lambda}),Q_i(\cdot|\lambda,\overline{\lambda})\rangle_{\omega}
\delta_{i,j},\quad\text{for $\lambda\in\mathbb{C}$ and $i,j\in\mathbb{Z}_{\geq 0}$}.\nonumber
\end{align}
Here, the subscripts $\lambda,\overline{\lambda}$ of $P,Q$ mean the conditioned point of the overlap. 
From the elementary linear algebra, we have
\[
\prod_{2\leq j,k\leq N}|z_j-z_k|^2\prod_{j=2}^N\omega(z_j,\overline{z}_j|z_1,\overline{z}_1)
=
\prod_{k=0}^{N-2}
\langle P_i(\cdot|\lambda,\overline{\lambda}),Q_i(\cdot|\lambda,\overline{\lambda})\rangle_{\omega}
\underset{{2\leq i,j\leq N}}{\det}\left(
K_{1,1}^{(N-1)}(z_i,\overline{z}_i,z_j,\overline{z}_j|z_1,\overline{z}_1)
\right),
\]
where $K_{11}^{(N)}$ is an integral kernel defined by
\[
K_{1,1}^{(N)}(z,\overline{z},w,\overline{w}|\lambda,\overline{\lambda})
=
\sum_{k=0}^{N-1}
\frac
{\overline{P_k(z|\lambda,\overline{\lambda})}Q_k(w|\lambda,\overline{\lambda})}
{\langle P_k(\cdot|\lambda,\overline{\lambda}),Q_k(\cdot|\lambda,\overline{\lambda})\rangle_{\omega}}
\omega(z,\overline{z}|\lambda,\overline{\lambda}).
\]
It will be convenient to define the reduced polynomial kernel $\cK^{(N)}$ via 
\begin{align}
K_{1,1}^{(N)}(z,\overline{z},w,\overline{w}|\lambda,\overline{\lambda})
:=&\mathcal{K}^{(N)}(\overline{z},w|\lambda,\overline{\lambda})\omega(z,\overline{z}|\lambda,\overline{\lambda}),
\nonumber
\\
\mathcal{K}_{1,1}^{(N)}(\overline{z},w|\lambda,\overline{\lambda}):=&\sum_{k=0}^{N-1}
\frac
{\overline{P_k(z|\lambda,\overline{\lambda})}Q_k(w|\lambda,\overline{\lambda})}
{\langle P_k(\cdot|\lambda,\overline{\lambda}),Q_k(\cdot|\lambda,\overline{\lambda})\rangle_{\omega}}.
\label{PolyK}
\end{align}
Then, we see that \eqref{eq:D11} is given by
\begin{equation}
\label{D11d}
D_{1,1}^{(N,k)}(\boldsymbol{z}_{(k)})=
\frac{N!}{Z_N}|z_1|^{2\alpha}e^{-|z_1|^2}
\prod_{k=0}^{N-2}
\langle P_i(\cdot|z_1,\overline{z_1}),Q_i(\cdot|z_1,\overline{z_1})\rangle_{\omega}
\underset{2\leq i,j\leq N}\det\left(
K_{1,1}^{(N-1)}(z_i,\overline{z}_i,z_j,\overline{z}_j|z_1,\overline{z}_1)
\right),
\end{equation}
and using Lemma \ref{thm:lem1}, we have 
\begin{align}
\label{D12d}
D_{1,2}^{(N,k)}(\boldsymbol{z}_{(k)})
=&
-\frac{N!}{Z_N}|z_1|^{2\alpha}|z_2|^{2\alpha}e^{-(|z_1|^2+|z_2|^2)}\cK_{1,1}^{(N-1)}(\zbar_1,z_2|z_1,\zbar_2)
\widehat{T}\left(\prod_{k=0}^{N-2}\langle P_k,Q_k\rangle_{\omega}\right)\\
&
\times\underset{3\leq i,j\leq N}\det\left(
\frac{\omega(z_i,\zbar_i|z_1,\zbar_2)}{\cK^{(N-1)}(\zbar_1,z_2|z_1,\zbar_2)}
\det\begin{pmatrix}
\cK_{1,1}^{(N-1)}(\zbar_1,z_2|z_1,\zbar_2) & \cK_{1,1}^{(N-1)}(\zbar_1,z_j|z_1,\zbar_2)\\
\cK_{1,1}^{(N-1)}(\zbar_i,z_2|z_1,\zbar_2) & \cK_{1,1}^{(N-1)}(\zbar_i,z_j|z_1,\zbar_2)
\end{pmatrix}
\right).
\nonumber
\end{align}
For the detailed deviations of \eqref{D11d} and \eqref{D12d}, we refer to \cite[p.13]{ATTZ}. 
In order to show an explicit form of \eqref{PolyK}, we introduce some functions here. 
We denote the truncated generalized exponential polynomial defined on $\C$ by 
\begin{equation}
e_n^{(\alpha)}(x):=\sum_{k=0}^n\frac{x^k}{\Gamma(k+\alpha+1)},\quad \text{for $\alpha>-1$}. 
\end{equation}
We also denote 
\begin{equation}
\mathfrak{e}_n^{(\alpha)}(z|x):=
e_n^{(\alpha)}(z)+\frac{1}{\Gamma(\alpha)}\frac{1}{x-\alpha},\quad\text{for $z,x\in\C$}.
\end{equation}
We write 
\begin{equation}
\varpi(z,w):=(zw)^{\alpha}e^{-\frac{1}{2}(|z|^2+|w|^2)},
\end{equation}
\begin{equation}
f_{n}^{(\alpha)}(x):=\frac{x-\alpha}{x}
\left\{
(n+\alpha+1)\mathfrak{e}_n^{(\alpha)}(x|x)
-x\mathfrak{e}_{n-1}^{(\alpha)}(x|x)
\right\}. 
\label{eq:fa}
\end{equation}
\begin{lem}\label{ConsOP}
Let 
\begin{equation}
G_{N-1}^{(\alpha)}(x|y,z):=\sum_{n,m=0}^{N-1}f_n^{(\alpha)}(x)f_m^{(\alpha)}(x)y^nz^m\sum_{k=\max\{n,m\}}^{N-1}
\frac{1}{\Gamma(k+\alpha+2)}\frac{x^k}{f_k^{(\alpha)}(x)f_{k+1}^{(\alpha)}(x)},\label{eq:PolynomialKernel1}
\end{equation}
where $f_{p}^{(\alpha)}(x)$ is defined in \eqref{eq:fa}
with $f_0^{(\alpha)}(x)=\frac{1}{\Gamma(\alpha+1)}$ and we understand $\sum_{k=0}^{-1}=0$. Then, we have
\begin{equation}
\cK_{1,1}^{(N)}(\overline{z},w|\lambda,\overline{\lambda})=
G_{N-1}^{(\alpha)}\left(\lambda\overline{\lambda}\Bigr|\frac{\zbar}{\overline{\lambda}},\frac{w}{\lambda}\right).\label{eq:PolynomialKernel2}
\end{equation}
\end{lem}
This finite $N$-kernel is not appropriate to take a large $N$-limit since \eqref{eq:PolynomialKernel2} contains the double summation and it seems to be the complicated form. 
To overcome this difficulty, we have to simplify this finite $N$-kernel.
The below is the building block in this paper. 
\begin{thm}\label{thm:Thm1}
For $n\in\N$, we define 
\begin{align}
\mathrm{I}_n^{(\alpha)}(\overline{z},w|\overline{\lambda},\lambda)
=&
\Bigl\{
\mathfrak{e}_n^{(\alpha)}(\lambda\overline{z}|\lambda\overline{\lambda})
\mathfrak{e}_n^{(\alpha)}(w\overline{\lambda}|\lambda\overline{\lambda})\\
&
-(1-(\zbar-\overline{\lambda})(w-\lambda))
\mathfrak{e}_n^{(\alpha)}(w\overline{z}|\lambda\overline{\lambda})
\mathfrak{e}_n^{(\alpha)}(\lambda\overline{\lambda}|\lambda\overline{\lambda})
\Bigr\}
\varpi(\overline{z},\lambda)\varpi(\overline{\lambda},w),
\label{KI}
\nonumber
\end{align}
\begin{equation}
\mathfrak{I}_{n}^{(\alpha)}(\overline{z},w|\overline{\lambda},\lambda)
=
\frac{\lambda\overline{\lambda}-\alpha}{\lambda\overline{\lambda}}
\frac{(n+\alpha+1)
\mathrm{I}_{n+1}^{(\alpha)}(\overline{z},w|\overline{\lambda},\lambda)
-\lambda\overline{\lambda}
\mathrm{I}_{n}^{(\alpha)}(\overline{z},w|\overline{\lambda},\lambda)
}
{f_{n}^{(\alpha)}(\lambda\overline{\lambda})\varpi(\overline{z},w)\varpi(\overline{\lambda},\lambda)
(\zbar-\overline{\lambda})^2(w-\lambda)^2},
\label{KIA}
\end{equation}
\begin{equation}
\mathfrak{II}_n^{(\alpha)}(\overline{z},w|\overline{\lambda},\lambda)
=
-\frac{\lambda\overline{\lambda}-\alpha}{\lambda\overline{\lambda}}\frac{1}{(\zbar-\overline{\lambda})(w-\lambda)}\frac{(\zbar w)^{n+1}}{\Gamma(n+\alpha+1)}
\frac{\mathfrak{e}_{n+1}^{(\alpha)}(\lambda\overline{\lambda}|\lambda\overline{\lambda})}
{f_n^{(\alpha)}(\lambda\overline{\lambda})},
\label{KII}
\end{equation}
and 
\begin{equation}
\mathfrak{III}_n^{(\alpha)}(\overline{z},w|\overline{\lambda},\lambda)
=
-\frac{1}{\Gamma(\alpha)}\frac{1}{\lambda\overline{\lambda}-\alpha}
\frac{1}{(\overline{z}-\overline{\lambda})(w-\lambda)}.
\label{KIII}
\end{equation}
Then, for any $2\leq k\leq N-1$, we have
\begin{equation}
D_{1,1}^{(N,k)}(\boldsymbol{z}_{(k)})=
f_{N-1}^{(\alpha)}(z_1\overline{z_1})
\varpi(\overline{z_1},z_1)
\underset{2\leq i,j\leq k}{\det}\left(K_{1,1}^{(N-1)}\left(z_i,\overline{z_i},z_j,\overline{z_j}|z_1,\overline{z_1}\right)\right),
\label{eq:ThmD11}
\end{equation}
where the correlation kernel $K_{1,1}^{(N)}$ is given by
\begin{equation}
 \label{SK1}
K_{1,1}^{(N)}(z,\zbar,w,\wbar|\lambda,\overline{\lambda})
=
\cK_{1,1}^{(N)}(\zbar,w|\lambda,\overline{\lambda})
\omega(\overline{z},z|\overline{\lambda},\lambda).
\end{equation}
Here, $\omega(\overline{z},z|\overline{\lambda},\lambda)$ is defined by \eqref{eq:weight} and
$\cK_{1,1}^{(N)}$ can be written in terms of \eqref{KIA}, \eqref{KII}, and \eqref{KIII}, i.e., 
\begin{equation}
 \label{SK2}
\cK_{1,1}^{(N)}(\zbar,w|\lambda,\overline{\lambda})
=
\mathfrak{I}_{N}^{(\alpha)}(\overline{z},w|\overline{\lambda},\lambda)
+
\mathfrak{II}_N^{(\alpha)}(\overline{z},w|\overline{\lambda},\lambda)
+
\mathfrak{III}_N^{(\alpha)}(\overline{z},w|\overline{\lambda},\lambda).
\end{equation}
Furthermore, for $3\leq k\leq N-1$, we have
\begin{align}
\label{eq:ThmD12}
D_{1,2}^{(N,k)}(\boldsymbol{z}_{(k)})
=
-f_{N-1}^{(\alpha)}(z_1\overline{z_2})
\varpi(z_1,\overline{z_2})
\varpi(\overline{z_1},z_2)
\cK_{1,1}^{(N-1)}(\overline{z_1},z_2|z_1,\overline{z_2})
\det_{3\leq i,j\leq k}\left(K_{1,2}^{(N-1)}\left(z_i,\zbar_i,z_j,\zbar_j|z_1,\overline{z_1},z_2,\overline{z_2}\right)\right),
\end{align}
where
\begin{equation}
\label{Kernel12}
K_{1,2}^{(N)}\left(z,\zbar,w,\wbar|u,\overline{u},v,\overline{v}\right)
=
\frac{\omega(z,\overline{z}|u,\overline{v})}{\cK_{1,1}^{(N)}(\overline{u},v|u,\overline{v})} 
\det
\begin{pmatrix}
\cK_{1,1}^{(N)}\left(\overline{u},v|u,\overline{v}\right) & \cK_{1,1}^{(N)}\left(\overline{u},w|u,\overline{v}\right)\\
\cK_{1,1}^{(N)}\left(\zbar,v|u,\overline{v}\right) & \cK_{1,1}^{(N)}\left(\zbar,w|u,\overline{v}\right)
\end{pmatrix}.
\end{equation}
\end{thm}
\begin{rem}
As already mentioned, this result for $\alpha=0$ was shown in \cite{ATTZ}.
Hence, our results should be regarded as the generalization of the case of the Ginibre unitary ensemble. 
Indeed, we can recover the case of the Ginibre unitary ensemble in \cite{ATTZ} when we take $\alpha=0$. 
\end{rem}
\begin{rem}
Interestingly, the form of the kernel \eqref{Kernel12} is same as the form of the conditional correlation kernel for the Palm measure of the determinantal point processes \cite{ST1}. 
Note that for the off-diagonal overlap case, the positivity as the point process is not guaranteed.
Hence, we can interpret that the decoupling operation in Lemma \ref{thm:lem1} essentially corresponds to the operation to take a Palm measure. 
\end{rem}
Once established the finite $N$-kernel \eqref{SK1}, 
we can obtain the scaling limits of the kernel in three regimes, namely, the strongly non-unitary regime, the weakly non-unitary regime, and at the singular origin. 
The below is our main result in this paper. 
\begin{thm}\label{MainThmS}
For $j=1,2,...,k$, let 
\begin{equation}
\text{
$
z_j=
\begin{cases}
e^{i\theta}\bigl(\sqrt{N}p+\zeta_j\bigr) 
& \text{if $p\in \mathrm{int}(S_{\mathrm{reg}})$ and in the strongly non-unitary regime}, \\
e^{i\theta}\bigl(\sqrt{N(1+b)}+\zeta_j\bigr) 
& \text{if at the outer edge and in the strongly non-unitary regime}, \\
e^{i\theta}\bigl(\sqrt{Nb}-\zeta_j\bigr) 
& \text{if at the inner edge and in the strongly non-unitary regime}, \\
e^{i\theta}\bigl(\sqrt{Na_N}+\zeta_j\bigr) 
& \text{if in the weakly non-unitary regime}, \\
\zeta_j
& \text{at the singular origin}. 
\end{cases}$}
\label{Zscaling}
\end{equation}
We define the rescaled $k$-th correlation function weighted by the on- or off-diagonal overlap
\begin{equation}
\mathfrak{D}_{1,1}^{(\ast,k)}(\boldsymbol{\zeta}_{(k)})
:=
\lim_{N\to\infty}
\begin{cases}
N^{-1}D_{1,1}^{(\mathrm{bulk})}(\boldsymbol{z}_{(k)}) & \text{if the bulk case and in the strongly non-unitary regime},\\
N^{-1/2}D_{1,1}^{(\mathrm{edge})}(\boldsymbol{z}_{(k)}) &\text{if the outer or inner edge cases, and in the strongly non-unitary regime}, \\
D_{1,1}^{(\mathrm{weak})}(\boldsymbol{z}_{(k)}) & \text{if the weakly non-unitary regime},\\
N^{-1}D_{1,1}^{(\mathrm{sing})}(\boldsymbol{z}_{(k)}) & \text{if at the singular origin}.
\end{cases}
\end{equation}
Then, we have
\begin{equation}
\mathfrak{D}_{1,1}^{(\ast,k)}(\boldsymbol{\zeta}_{(k)})
=
\Psi_{1,1}^{(\ast)}(\zeta_1,\overline{\zeta_1})
\underset{2\leq i,j\leq k}{\det}\Bigl(K_{1,1}^{(\ast)}(\zeta_i,\overline{\zeta_i},\zeta_j,\overline{\zeta_j}|\zeta_1,\overline{\zeta_1})\Bigr),
\end{equation}
where the convergence is uniform for $\zeta_j$ in compact subsets of $\C$ in any cases. 
Here, $K_{1,1}^{(\ast)}$ is given by
\begin{equation}
K_{1,1}^{(\ast)}(\overline{\zeta},\zeta,\overline{\eta},\eta,|\overline{\chi},\chi)
=
\cK_{1,1}^{(\ast)}(\overline{\zeta},\eta|\overline{\chi},\chi)\omega^{(\ast)}(\overline{\zeta},\zeta|\overline{\chi},\chi)\quad\text{for $\ast\in\{\mathrm{bulk},\mathrm{edge},\mathrm{weak},\mathrm{sing}\}$}
\end{equation}
where $\cK_{1,1}^{(\ast)}(\overline{\zeta},\eta|\overline{\chi},\chi)$, $\omega^{(\ast)}(\overline{\zeta},\zeta|\overline{\chi},\chi)$, and $\Psi_{1,1}^{(\ast)}(\overline{\zeta},\zeta)$ change depending on $\{\mathrm{bulk},\mathrm{edge},\mathrm{weak},\mathrm{sing}\}$.
Moreover, for the off-diagonal case, we have 
\begin{equation}
\mathfrak{D}_{1,2}^{(\ast,k)}(\boldsymbol{\zeta}_{(k)})
=
\Psi_{1,2}^{(\ast)}(\zeta_1,\overline{\zeta_2})
\underset{3\leq i,j\leq k}{\det}\Bigl(K_{1,2}^{(\ast)}(\zeta_i,\overline{\zeta_i},\zeta_j,\overline{\zeta_j}|\zeta_1,\overline{\zeta_1},\zeta_2,\overline{\zeta_2})\Bigr),
\end{equation}
where the convergence is uniform for $\zeta_j$ in compact subsets of $\C$ in any cases, and $K_{1,2}^{(\ast)}$ is given by
\begin{equation}
K_{1,2}^{(\ast)}=\frac{\omega^{(\ast)}(\zeta_i,\overline{\zeta_i}|\zeta_1,\overline{\zeta_2})}{\cK_{1,1}^{(\ast)}(\overline{\zeta_1},\zeta_2|\zeta_1,\overline{\zeta_2})}
\det
\begin{pmatrix}
\cK_{1,1}^{(\ast)}(\overline{\zeta_1},\zeta_2|\zeta_1,\overline{\zeta_2}) & 
\cK_{1,1}^{(\ast)}(\overline{\zeta_1},\zeta_j|\zeta_1,\overline{\zeta_2}) \\
\cK_{1,1}^{(\ast)}(\overline{\zeta_i},\zeta_2|\zeta_1,\overline{\zeta_2}) & 
\cK_{1,1}^{(\ast)}(\overline{\zeta_i},\zeta_j|\zeta_1,\overline{\zeta_2})
\end{pmatrix}.
\end{equation}
In particular, $\cK_{1,1}^{(\ast)}(\overline{\zeta},\eta|\overline{\chi},\chi)$, $\omega^{(\ast)}(\overline{\zeta},\zeta|\overline{\chi},\chi)$, $\Psi_{1,1}^{(\ast)}(\overline{\zeta},\zeta)$, and $\Psi_{1,2}^{(\ast)}(\overline{\zeta},\zeta)$ are given as follows:
\begin{itemize}
\item[\textup{(i)}] \textbf{\textup{Strongly non-unitary regime (bulk case):}}\label{StrongBulkCase}
we have
\begin{equation}
\label{BulkPsi}
\Psi_{1,1}^{(\mathrm{bulk})}(\zeta_1)
=
\Psi_{1,2}^{(\mathrm{bulk})}(\zeta_1,\zeta_2)
=
\frac{(|p|^2-b)(1+b-|p|^2)}{|p|^2}\quad\text{for $p\in\mathrm{int}(S_{\mathrm{reg}})$},
\end{equation}
\begin{equation}
\cK_{1,1}^{(\mathrm{bulk})}\left(\overline{\zeta},\eta|\chi,\overline{\chi}\right)=
\left.\frac{d}{dt}\left(\frac{e^t-1}{t}\right)\right|_{t=(\overline{\zeta}-\overline{\chi})(\eta-\chi)},
\end{equation}
and 
\begin{equation}
\omega^{(\mathrm{bulk})}\left(\zeta,\overline{\zeta}|\chi,\overline{\chi}\right)
=
\left(1+(\zeta-\chi)(\overline{\zeta}-\overline{\chi})\right)e^{-(\zeta-\chi)(\overline{\zeta}-\overline{\chi})}.
\end{equation}
\item[\textup{(ii)}] \textbf{\textup{Strongly non-unitary regime (edge case):}}\label{StrongEdgeCase}
For $a,b,c,d,f\in\C$, we write
\begin{equation}
\label{eq:Hf}
H(a,b,c,d,f)
:=-\frac{e^{-\frac{1}{2}a^2}\left.\frac{d}{dx}
\left[
e^{\frac{(a+x)^2}{2}}
\left(
e^{-f}F(b+x)F(c+x)-F(d+x)F(a+x)+fF(d)F(a+x)
\right)
\right]\right|_{x=0}}{\cF(a)},
\end{equation}
where
\begin{equation}
\cF(a)=
\frac{e^{-\frac{1}{2}a^2}}{\sqrt{2\pi}}\Bigl(1-\sqrt{2\pi}ae^{\frac{1}{2}a^2}F(a)\Bigr)
\quad
 \text{for $a\in\C$}.
\end{equation}
We also write 
\begin{equation}
c_b=
\begin{cases}
(1+b)^{-1/2} & \text{if the outer edge case}, \\
b^{-1/2} & \text{if the inner edge case}.
\end{cases}
\label{ConstCb}
\end{equation}
Then, we have 
\begin{align}
\Psi_{1,1}^{(\mathrm{edge})}(\zeta_1)
=&
c_b\cF(\zeta_1+\overline{\zeta_1}),
\label{EdgePsi11}
\\
\Psi_{1,2}^{(\mathrm{edge})}(\zeta_1,\zeta_2)
=&
c_be^{-|\zeta_1-\zeta_2|^2}
\cF(\zeta_1+\overline{\zeta_2}) 
\frac{H\left(\zeta_1+\overline{\zeta}_2,\zeta_1+\overline{\zeta}_1,\zeta_2+\overline{\zeta}_2,\zeta_2+\overline{\zeta}_1,-|\zeta_1-\zeta_2|^2\right)}{(\overline{\zeta}_1-\overline{\zeta}_2)^2(\zeta_1-\zeta_2)^2},
\label{EdgePsi12}
\end{align}
\begin{equation}
\cK_{1,1}^{(\mathrm{edge})}\left(\overline{\zeta},\eta|\chi,\overline{\chi}\right)=
e^{\overline{\zeta}\eta}
\frac{H(\chi+\overline{\chi},\chi+\overline{\zeta},\overline{\chi}+\eta,\overline{\zeta}+\eta,(\overline{\zeta}-\overline{\chi})(\eta-\chi))}{(\overline{\zeta}-\overline{\chi})^2(\eta-\chi)^2},
\end{equation}
and
\begin{equation}
\omega^{(\mathrm{edge})}\left(\zeta,\overline{\zeta}|\chi,\overline{\chi}\right)
=
\left(1+(\zeta-\chi)(\overline{\zeta}-\overline{\chi})\right)e^{-\zeta\overline{\zeta}}.
\end{equation}
\item[\textup{(iii)}] \textbf{\textup{Weakly non-unitary regime:}}\label{WeakCase}
For $x\in\C$ and $\rho>0$, we define
\begin{equation}
\label{cL}
\cL_{\rho}(x)
=
\frac{1}{\sqrt{2\pi}}
\Bigl\{
\left(x+\frac{\rho}{2}\right)e^{-\frac{1}{2}\left(x-\frac{\rho}{2}\right)^2}
-
\left(x-\frac{\rho}{2}\right)
\left(\sqrt{2\pi}\left(x+\frac{\rho}{2}\right)L_{\rho}(x)+e^{-\frac{1}{2}\left(x+\frac{\rho}{2}\right)^2}\right)\Bigr\}.
\end{equation}
We write 
\[
\mathcal{A}_{\rho}(a,b,c,d,f)
=
\Bigl(\frac{\rho}{2}+a\Bigr)\Bigl(\frac{\rho}{2}-a\Bigr)
\Bigl(
e^{f}(f-1)L_{\rho}(a)L_{\rho}(d)+L_{\rho}(b)L_{\rho}(c)
\Bigr),
\]
\[
\mathcal{B}_{\rho}(a,b,c)
=
L_{\rho}(b)\Bigl(
\Bigl(a+\frac{\rho}{2}\Bigr)\frac{e^{-\frac{1}{2}(c-\frac{\rho}{2})^2}}{\sqrt{2\pi}}
-
\Bigl(a-\frac{\rho}{2}\Bigr)\frac{e^{-\frac{1}{2}(c+\frac{\rho}{2})^2}}{\sqrt{2\pi}}
\Bigr),
\]
\[
\cC_{\rho}(a,b)
=
\frac{e^{-\frac{1}{2}(a+\frac{\rho}{2})^2-\frac{1}{2}(b-\frac{\rho}{2})^2}}{\sqrt{2\pi}^2}
+
\frac{e^{-\frac{1}{2}(a-\frac{\rho}{2})^2-\frac{1}{2}(b+\frac{\rho}{2})^2}}{\sqrt{2\pi}^2}
\]
and
\begin{align}
\label{cH}
\mathcal{H}_{\rho}(a,b,c,d,f)
=&
\mathcal{A}_{\rho}(a,b,c,d,f)
+
e^{f}f\mathcal{B}_{\rho}(a,d,a)
+
\mathcal{B}_{\rho}(a,b,c)+\mathcal{B}_{\rho}(a,c,b)
\\
&
-
e^{f}\mathcal{B}_{\rho}(a,d,a)
-e^{f}\mathcal{B}_{\rho}(a,a,d)
+
\cC_{\rho}(b,c)-e^{f}\cC_{\rho}(d,a).
\nonumber
\end{align}
Here, we recall that $L_{\rho}$ is \eqref{banderror}. Then, we have
\begin{equation}
\Psi_{1,1}^{(\mathrm{weak})}(\zeta_1)=\cL_{\rho}(\zeta_1+\overline{\zeta_1}),
\end{equation}
\begin{equation}
\Psi_{1,2}^{(\mathrm{weak})}(\zeta_1,\zeta_2)=
e^{-|\zeta_1-\zeta_2|^2}
\frac{\cH_{\rho}(\zeta_1+\overline{\zeta_2},\zeta_1+\overline{\zeta_1},\zeta_2+\overline{\zeta_2},\zeta_2+\overline{\zeta_1},-(\overline{\zeta_1}-\overline{\zeta_2})(\zeta_1-\zeta_2))}
{(\overline{\zeta_1}-\overline{\zeta_2})^2(\zeta_1-\zeta_2)^2},
\end{equation}
\begin{equation}
\cK^{(\mathrm{weak})}\left(\overline{\zeta},\eta|\chi,\overline{\chi}\right)
=
\frac{\cH_{\rho}(\chi+\overline{\chi},\chi+\overline{\zeta},\eta+\overline{\chi},\eta+\overline{\zeta},(\overline{\zeta}-\overline{\chi})(\eta-\chi))}{\cL_{\rho}(\chi+\overline{\chi})(\overline{\zeta}-\overline{\chi})^2(\eta-\chi)^2},
\end{equation}
and
\begin{equation}
\omega^{(\mathrm{weak})}\left(\zeta,\overline{\zeta}|\chi,\overline{\chi}\right)
=
\omega^{(\mathrm{bulk})}\left(\zeta,\overline{\zeta}|\chi,\overline{\chi}\right).
\end{equation}
\item[\textup{(iv)}] \textbf{\textup{At the singular origin:}}\label{SingularCase}
For $z,x\in\C$ and $b>0$, we define 
\begin{equation}
\label{cEb}
\cE_{b}(z|x)=(x-b)E_{1,b+1}(z)+\frac{1}{\Gamma(b)},
\end{equation}
and for $x,y,z,w,f\in\C$, we define
\begin{align}
\label{cS}
\cS_b(x,y,z,w,f)
=&
(x-b)
(E_{1,b+1}(y)E_{1,b+1}(z)-E_{1,b+1}(w)E_{1,b+1}(x)+fE_{1,b+1}(w)E_{1,b+1}(x))
\\
&
+\frac{E_{1,b+1}(y)+E_{1,b+1}(z)-E_{1,b+1}(w)-E_{1,b+1}(x)+fE_{1,b+1}(w)}{\Gamma(b)}.
\nonumber
\end{align}
Here, we recall that $E_{1,b+1}$ is \eqref{ML}.
Then, we have
\begin{equation}
\Psi_{1,1}^{(\mathrm{sing})}(\zeta_1)
=\cE_b(|\zeta_1|^2)|\zeta_1|^{2b-2}e^{-|\zeta_1|^2}.
\end{equation}
\begin{equation}
\Psi_{1,2}^{(\mathrm{sing})}(\zeta_1,\zeta_2)
=
\frac{\cS_b(\overline{\zeta_2}\zeta_1,\overline{\zeta_1}\zeta_1,\overline{\zeta_2}\zeta_2,\overline{\zeta_1}\zeta_2,(\overline{\zeta_1}-\overline{\zeta_2})(\zeta_2-\zeta_1))}{\zeta_1\overline{\zeta_2}(\overline{\zeta_1}-\overline{\zeta_2})^2(\zeta_2-\zeta_1)^2}|\zeta_1|^{2b}|\zeta_2|^{2b}e^{-(|\zeta_1|^2+|\zeta_2|^2)},
\end{equation}
\begin{equation}
\cK_{1,1}^{(\mathrm{sing})}(\overline{\zeta},\eta|\chi,\overline{\chi})
=
\frac{\cS_b(\overline{\chi}\chi,\overline{\zeta}\chi,\overline{\chi}\eta,\overline{\zeta}\eta(\overline{\zeta}-\overline{\chi})(\eta-\chi))}
{(\overline{\zeta}-\overline{\chi})^2(\eta-\chi)^2\cE_b(\chi\overline{\chi})},
\end{equation}
and
\begin{equation}
\omega^{(\mathrm{sing})}(\zeta,\overline{\zeta}|\chi,\overline{\chi})
=
\left(1+(\overline{\zeta}-\overline{\chi})(\zeta-\chi)\right)|\zeta|^{2b} e^{-|\zeta|^2}.
\end{equation}
\end{itemize}
\end{thm}
\begin{figure}[htbp]
  \begin{minipage}[b]{0.3\linewidth}
    \centering
    \includegraphics[keepaspectratio, scale=0.35]{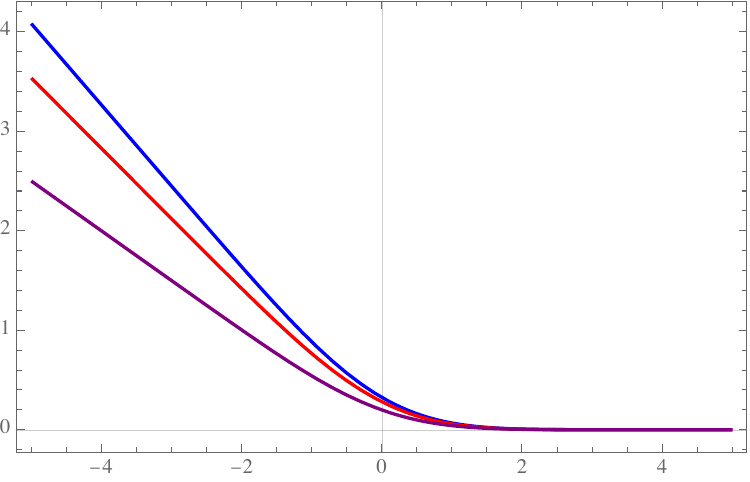}
    \subcaption{Graph of $\cF(x)$.}
  \end{minipage}
  \begin{minipage}[b]{0.3\linewidth}
    \centering
    \includegraphics[keepaspectratio, scale=0.35]{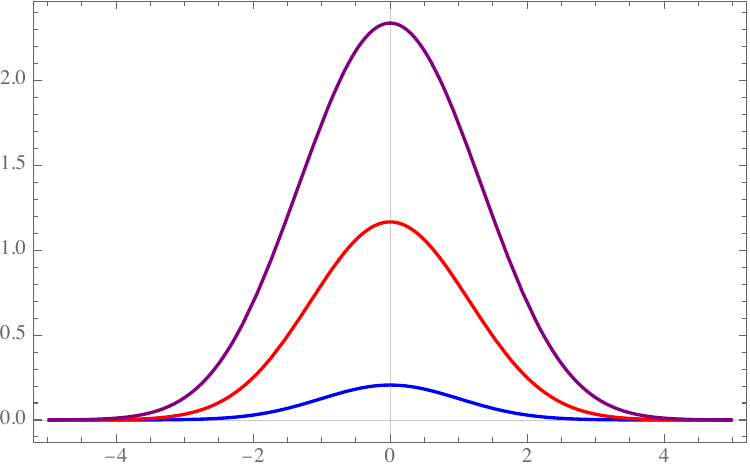}
    \subcaption{Graph of $\cL_{\rho}(x)$.}
  \end{minipage}
    \begin{minipage}[b]{0.3\linewidth}
    \centering
    \includegraphics[keepaspectratio, scale=0.35]{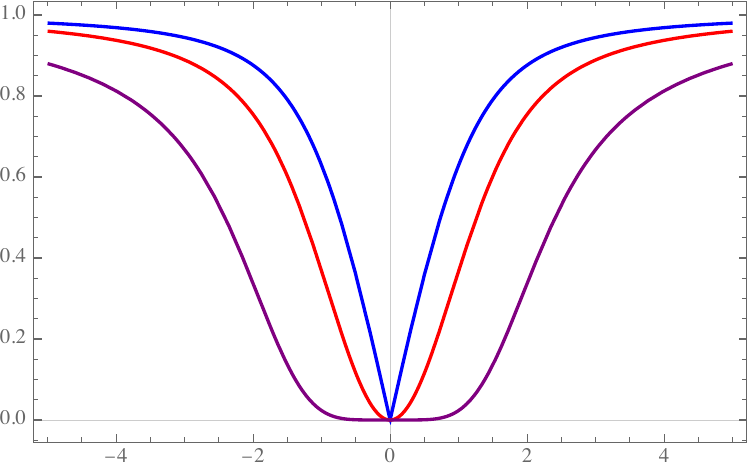}
    \subcaption{Graph of $\cE_b(x)$.}
  \end{minipage}
  \caption{In the figure (b), the blue line is $\rho=0.5$, the red line is $\rho=2$, and the purple line is $\rho=3$. In the figure $(c)$, the blue line is $b=0.5$, the red line is $b=1$, and the purple line is $b=3$. 
  }
\end{figure}
\begin{figure}[htbp]
  \begin{minipage}[b]{0.4\linewidth}
    \centering
    \includegraphics[keepaspectratio, scale=0.4]{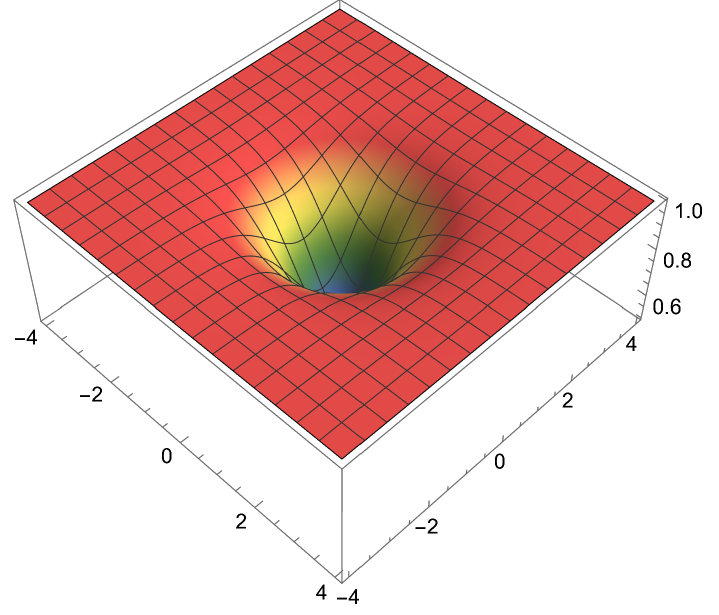}
    \subcaption{Graph of $(x,y)\mapsto K_{1,1}^{(\mathrm{bulk})}(x-iy,x+iy,x-iy,x+iy|\overline{\chi},\chi)$.}
  \end{minipage}
  \begin{minipage}[b]{0.4\linewidth}
    \centering
    \includegraphics[keepaspectratio, scale=0.4]{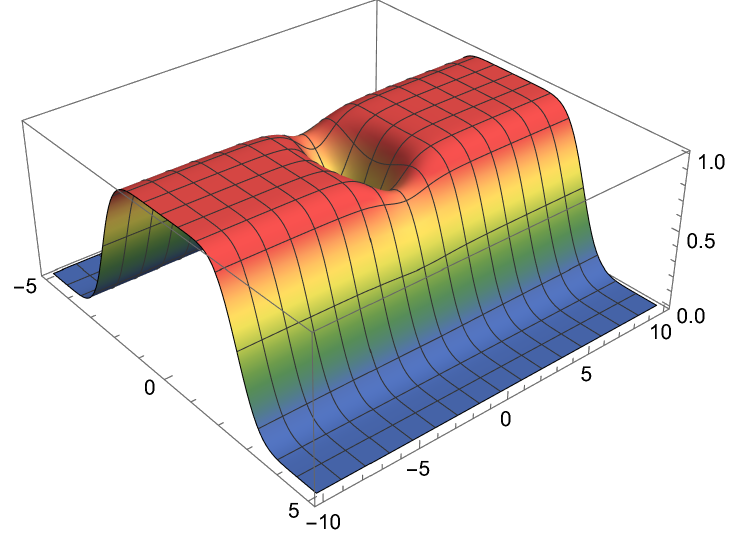}
    \subcaption{Graph of $(x,y)\mapsto K_{1,1}^{(\mathrm{weak})}(x-iy,x+iy,x-iy,x+iy|\overline{\chi},\chi)$.}
  \end{minipage}
  \caption{In the left figure, we set $\chi=0$. In the right figure, we set $\chi=0,\rho=10$. 
  }
\end{figure}
\begin{rem}[Universality in the strongly non-unitary regime]
The bulk and edge cases in strongly non-unitary regime in Theorem~\ref{MainThmS} show the universality of the weighted multi-intensity of the on- and off-diagonal overlaps. 
Indeed, their limiting kernels have already appeared in \cite{ATTZ}. 
The only differences with the results in \cite{ATTZ} are constants in \eqref{BulkPsi}, \eqref{EdgePsi11}, and \eqref{EdgePsi12}. 
These differences are due to missing the scaling factor in \cite{ATTZ}. 
Hence, we can conclude the universality of the scaling limits of the weighted multi-point intensity of the on- and off-diagonal overlaps in the strongly non-unitary regime. 
\end{rem}
\begin{rem}[Limits of $\rho\to\infty$ and $\rho\to0$ in the weakly non-unitary regime]
By the definition of $\cH_{\rho}$ and $\cL_{\rho}$, it is easy to see that as $\rho\to\infty$, 
\[
\cH_{\rho}(\chi+\overline{\chi},\chi+\overline{\zeta},\eta+\overline{\chi},\eta+\overline{\zeta},(\overline{\zeta}-\overline{\chi})(\eta-\chi))
\sim
\frac{\rho^2}{4}
\bigl(
e^{(\overline{\zeta}-\overline{\chi})(\eta-\chi)}(\overline{\zeta}-\overline{\chi})(\eta-\chi)
-e^{(\overline{\zeta}-\overline{\chi})(\eta-\chi)}
+1
\bigr),
\]
and 
\[
\cL_{\rho}(\chi+\overline{\chi})\sim
\frac{\rho^2}{4}.
\]
Hence, we find that 
\[
K_{1,1}^{(\mathrm{weak})}(\zeta,\overline{\zeta},\eta,\overline{\eta}|\overline{\chi},\chi)
\to
\frac{d}{dx}\left(\frac{e^x-1}{x}\right)\Bigr|_{x=(\overline{\zeta}-\overline{\chi})(\eta-\chi)}
\omega^{(\mathrm{bulk})}(\zeta,\overline{\zeta}|\chi,\overline{\chi})
=
K_{1,1}^{(\mathrm{bulk})}(\zeta,\overline{\zeta},\eta,\overline{\eta}|\overline{\chi},\chi)
\quad\text{as $\rho\to\infty$}.
\]
From the above discussions, we have 
\[
\frac{4}{\rho^2}\mathfrak{D}_{1,1}^{(\mathrm{weak})}(\boldsymbol{\zeta}_{(k)})
\to
\mathfrak{D}_{1,1}^{(\mathrm{weak})}(\boldsymbol{\zeta}_{(k)})
\quad\text{as $\rho\to\infty$}.
\]
Next, we consider the change of the variables $\zeta\mapsto \zeta/\alpha,\eta\mapsto \eta/\alpha,\chi\mapsto \chi/\alpha$ for $\alpha=\rho/2$.
Note that the contribution in the denominator is 
\[
\cL_{\rho}\Bigl(\frac{\chi+\overline{\chi}}{\alpha}\Bigr)
\sim
-(\chi+\overline{\chi})^2\frac{1}{\alpha^2}L_{\rho}\Bigl(\frac{\chi+\overline{\chi}}{\alpha}\Bigr)
\quad\text{as $\alpha\to0$}.
\]
Notice also that 
\begin{align*}
\cH_{\rho}
\Bigl(\frac{\chi+\overline{\chi}}{\alpha},\frac{\chi+\overline{\zeta}}{\alpha},\frac{\eta+\overline{\chi}}{\alpha},\frac{\eta+\overline{\zeta}}{\alpha},\frac{(\overline{\zeta}-\overline{\chi})(\eta-\chi)}{\alpha^2}\Bigr)
\sim&
-(\chi+\overline{\chi})^2
e^{\frac{(\overline{\zeta}-\overline{\chi})(\eta-\chi)}{\alpha^2}}
(\overline{\zeta}-\overline{\chi})(\eta-\chi)
\frac{L_{\rho}\Bigl(\frac{\eta+\overline{\zeta}}{\alpha}\Bigr)}{\alpha^2}
\frac{L_{\rho}\Bigl(\frac{\chi+\overline{\chi}}{\alpha}\Bigr)}{\alpha^2},
\end{align*}
and
\begin{align*}
\frac{\alpha^4}{(\overline{\zeta}-\overline{\chi})^2(\eta-\chi)^2}
\Bigl(1+\frac{(\overline{\zeta}-\overline{\chi})(\zeta-\chi)}{\alpha^2}\Bigr)e^{-\frac{(\overline{\zeta}-\overline{\chi})(\zeta-\chi)}{\alpha^2}}
=&
\frac{\alpha^2}{(\overline{\zeta}-\overline{\chi})^2(\eta-\chi)^2}
\Bigl(\alpha^2+(\overline{\zeta}-\overline{\chi})(\zeta-\chi)\Bigr)e^{-\frac{(\overline{\zeta}-\overline{\chi})(\zeta-\chi)}{\alpha^2}}
\\
\sim&
\frac{\alpha^2}{(\overline{\zeta}-\overline{\chi})(\eta-\chi)}
e^{-\frac{(\overline{\zeta}-\overline{\chi})(\zeta-\chi)}{\alpha^2}}
\quad\text{as $\alpha\to0$}.
\end{align*}
Hence, we see that as $\alpha\to0$,
\begin{align*}
\frac{1}{\alpha^2}K_{1,1}^{(\mathrm{weak})}\bigl(\frac{\zeta}{\alpha},\frac{\overline{\zeta}}{\alpha},\frac{\eta}{\alpha},\frac{\overline{\eta}}{\alpha}|\frac{\chi}{\alpha},\frac{\overline{\chi}}{\alpha}\bigr)
\sim
e^{\frac{(\overline{\zeta}-\overline{\chi})(\eta-\chi)}{\alpha^2}}
e^{-\frac{(\overline{\zeta}-\overline{\chi})(\zeta-\chi)}{\alpha^2}}
e^{\frac{1}{2\alpha^2}(|\zeta|^2+|\eta|^2)-\overline{\zeta}\eta}
e^{-\frac{1}{2\alpha^2}(|\zeta|^2+|\eta|^2)+\overline{\zeta}\eta}
\frac{L_{\rho}\Bigl(\frac{\eta+\overline{\zeta}}{\alpha}\Bigr)}{\alpha^2}. 
\end{align*}
Hence, there exists a co-cycle factor $\psi(\zeta)$, which does not affect the value of the determinant, such that as $\alpha\to0$, 
\[
\frac{1}{\alpha^2}K_{1,1}^{(\mathrm{weak})}\bigl(\frac{\zeta}{\alpha},\frac{\overline{\zeta}}{\alpha},\frac{\eta}{\alpha},\frac{\overline{\eta}}{\alpha}|\frac{\chi}{\alpha},\frac{\overline{\chi}}{\alpha}\bigr)
\sim
\psi(\eta)
\frac{1}{\pi}\frac{\sin(u-v)}{u-v}
\psi^{-1}(\zeta)
\]
with $u=\im(\zeta)$ and $v=\im(\eta)$. 
Therefore, after dividing the front factor $\Psi_{1,1}(\chi)=\cL_{\rho}(\chi+\overline{\chi})$ by itself in advance, we can find that 
\[
\frac{1}{\cL_{\rho}((\zeta_1+\overline{\zeta_1})/\alpha)}\frac{1}{\alpha^{2k}}\mathfrak{D}_{1,1}^{(\mathrm{weak})}(\boldsymbol{\zeta}_{(k)}/\alpha)\to\underset{2\leq i,j\leq k}{\det}\Bigl(\frac{1}{\pi}\frac{\sin(u_i-u_j)}{u_i-u_j}\Bigr)\quad\text{as $\alpha=\frac{\rho}{2}\to0$ with $u_i=\im(\zeta_i)$}.
\]
\end{rem}
$\Psi_{1,1}^{(\ast)}$ and $\Psi_{1,2}^{(\ast)}$ for $\ast\in\{\mathrm{bulk,edge,weak,sing}\}$in Theorem~\ref{MainThmS} correspond to the scaling limits of the on- and off-diagonal conditional expectation of the overlap, respectively. 
We summarize these facts as a corollary below, and we conclude this section. 
\begin{cor}
We have 
\begin{align*}
\lim_{N\to\infty}\frac{1}{Na_N}\E[\cO_{1,1}|z_1=\eqref{Zscaling}]=&\Psi_{1,1}^{(\ast)}(\zeta_1),\quad \ast\in\{\mathrm{bulk,edge,weak,sing}\},
\\
\lim_{N\to\infty}\frac{1}{Na_N}\E[\cO_{1,2}|z_1,z_2=\eqref{Zscaling}]=&\Psi_{1,2}^{(\ast)}(\zeta_1,\zeta_2),\quad \ast\in\{\mathrm{bulk,edge,weak,sing}\}.
\end{align*}
\end{cor}
\section{Finite $N$-kernel: Proof of Theorem \ref{thm:Thm1}}\label{s4}
In this section, we prove Theorem \ref{thm:Thm1}. 
Our strategy to derive the finite $N$-kernel follows the strategy of \cite{ATTZ}.
Its proof is divided by some steps. 
The first step is to construct the planar orthogonal polynomials with respect to the weight function \eqref{eq:weight} using the moment method. 
However, as we already mentioned, our finite $N$-kernel in that step has the complicated form, 
and it would not be suitable to show the scaling limits. 
As the second step, we manage to simplify the finite $N$-kernel to take scaling limits. 
With help of these steps, we finally conclude the proof of Theorem~\ref{thm:Thm1}.  
We emphasize that the way to derive the finite $N$-kernel for the $k$-th conditional expectation of the overlaps outlined here is robust for radially symmetric potentials such as the induced spherical ensemble and the truncated unitary ensemble. 
\subsection{Proof of Lemma~\ref{ConsOP}}\label{s2-1}
In this subsection, we prove Lemma~\ref{ConsOP} based on the moment method. 
\begin{proof}[Proof of Lemma~\ref{ConsOP}]
First, we define the moment matrix $M$ with entries $M_{i,j}:=\langle z^i, z^j\rangle_{\omega}$ with respect to the inner product \eqref{innerP}. Then, we have
\begin{equation}
M_{i,j}
=\Gamma(i+\alpha+1)
\left\{
\delta_{i,j}(\alpha+i+2+|\lambda|^2)-\delta_{i+1,j}\overline{\lambda}(i+\alpha+1)-\delta_{i,j+1}\lambda
\right\}\\
=:\Gamma(i+\alpha+1)
\mu_{i,j}.
\end{equation}
As in \cite{ATTZ}, we shall perform the LDU decomposition of the matrix $\mu=\left(\mu_{i,j}\right)_{i,j}$. If we have the LDU decomposition as $\mu=LDU$, where 
\begin{equation}
D_{p,q}=d_p\delta_{p,q},\quad L_{p,q}=\delta_{p,q}+\ell_{p}\delta_{p,q+1},\quad U_{p,q}=\delta_{p,q}+u_{q}\delta_{q,p+1},\quad \text{for $p,q\geq 0$},\label{eq:LDU1}
\end{equation}
then we see that 
\begin{equation}
d_p=-d_{p-1}\ell_{p-1}u_{p}\mathbf{1}_{p\geq 1}+2+\alpha+\lambda\overline{\lambda}+p,
\quad
u_{p+1}=-\frac{\lambda(p+\alpha+1)}{d_p}.
\quad
\ell_{p+1}=-\frac{\overline{\lambda}}{d_p}, \label{eq:LDU2}
\end{equation}
Hence, we have the following recurrence equation:
\[
d_p=2+\alpha+\lambda\overline{\lambda}+p-\frac{\lambda\overline{\lambda}(p+\alpha)}{d_{p-1}},
\quad
p\geq 1,
\quad
d_0=\alpha+2+\lambda\overline{\lambda}. 
\]
Let $x=|\lambda|^2$ and $d_p=\frac{r_{p+1}}{r_p}$ with $r_0=1$. Then, we have
\begin{equation}
r_{p+1}+(p+\alpha)xr_{p-1}=(2+\alpha+x+p)r_p,\quad
r_1=d_0=\alpha+2+x. \label{eq:RUC1}
\end{equation}
From an induction argument, we find that the unique solution of the recurrence equation \eqref{eq:RUC1} is given by
\[
r_p=\mathcal{D}_{p-1}=\Gamma(p+\alpha+1)h_{p}^{(\alpha)}(x),
\]
where
\begin{equation}
h_{p}^{(\alpha)}(x)
=\frac{x-\alpha}{x}\sum_{k=0}^p\frac{p+1-k}{\Gamma(k+\alpha+1)}x^k+\frac{\alpha(p+1)}{\Gamma(\alpha+1)}\frac{1}{x}
\quad\text{with $h_0^{(\alpha)}(x)=\frac{1}{\Gamma(\alpha+1)}$}.
 \label{eq:f1}
\end{equation}
Before completing the proof, it is convenient to simplify $h_p^{(\alpha)}(x)$. From \eqref{eq:f1}, 
\begin{align*}
h_p^{(\alpha)}(x)
=&
\frac{x-\alpha}{x}(p+1)\sum_{k=0}^p\frac{x^k}{\Gamma(k+\alpha+1)}
-
(x-\alpha)\sum_{k=0}^p\frac{kx^{k-1}}{\Gamma(k+\alpha+1)}
+
\frac{1}{x}\frac{\alpha(p+1)}{\Gamma(\alpha+1)}\\
=&
\frac{x-\alpha}{x}\left\{(p+1)e_{p}^{(\alpha)}(x)-xe_{p-1}^{(\alpha)}(x)+\alpha xe_{p-1}^{(\alpha+1)}(x)+\frac{\alpha}{x-\alpha}\frac{p+1}{\Gamma(\alpha+1)}\right\},
\end{align*}
Here, recall that $e_{N-1}^{(\alpha)}(x)$ can be written in terms of the incomplete Gamma function:
\begin{equation}
e_{N-1}^{(\alpha)}(z)
=\frac{e^z}{z^\alpha}\left(Q(N+\alpha,z)-Q(\alpha,z)\right). \label{eq:incompleteG}
\end{equation}
Using \eqref{eq:incompleteG}, 
we see that 
\[
h_N^{(\alpha)}(x)
=\frac{x-\alpha}{x}
\left(
(N+\alpha+1)e_N^{(\alpha)}(x)-xe_{N-1}^{(\alpha)}(x)
+\frac{\alpha}{\Gamma(\alpha+1)}\frac{N+\alpha+1-x}{x-\alpha}
\right)
=f_{N}^{(\alpha)}(x),
\]
where $f_N^{(\alpha)}(x)$ is defined as \eqref{eq:fa}.
From now on, we use the notation \eqref{eq:fa}.
Coming back to the LDU decomposition step and from the definition of $d_p$, we obtain 
\[
d_p=\frac{r_{p+1}}{r_p}=\frac{\Gamma(p+\alpha+2)f_{p+1}^{(\alpha)}(x)}{\Gamma(p+\alpha+1)f_{p}^{(\alpha)}(x)}
=(p+\alpha+1)\frac{f_{p+1}^{(\alpha)}(x)}{f_{p}^{(\alpha)}(x)}.
\]
Recall that
$\langle P_k(\cdot|\lambda,\overline{\lambda}),Q_\ell(\cdot|\lambda,\overline{\lambda})\rangle_{\omega} =D_{kk}\delta_{k,\ell},$ 
$
P_{k}(z|\lambda,\overline{\lambda})=\sum_{m=0}^k(\overline{L}^{-1})_{km}z^m,
$
and 
$
Q_{k}(z|\lambda,\overline{\lambda})=\sum_{m=0}^kz^m(U^{-1})_{mk}
$
for $k\in\Z_{\geq 0}$. 
Then, we have
$
\cK_{11}^{(N)}(z,\zbar|\lambda,\overline{\lambda})=\sum_{i,j=0}^{N-1}z^iC_{i,j}^{(N-1)}\zbar^j,
$
where
$
C_{i,j}^{(N)}=\sum_{k=0}^{N}(U^{-1})_{ik}\frac{1}{D_{kk}}(L^{-1})_{kj}
$
for $i,j\geq 0$.
From \eqref{eq:LDU1} and \eqref{eq:LDU2}, multiplying the diagonal matrix $\diag\left(\Gamma(i+\alpha+1)\right)$ by $\mu=LDU$ and updating notations, we have
\[
L_{pm}=\delta_{pm}-\overline{\lambda}
\frac{f_{p-1}^{(\alpha)}(x)}{f_{p}^{(\alpha)}(x)}
\delta_{p,m+1},\quad
U_{mq}=\delta_{mq}-\lambda
\frac{f_{m-1}^{(\alpha)}(x)}{f_{m}^{(\alpha)}(x)}
\delta_{q,m+1},\quad (m,q\geq0),
\]
\[
D_{m}=\langle P_m,Q_m\rangle_{\omega_\alpha}=\Gamma(m+\alpha+2)
\frac{f_{m+1}^{(\alpha)}(x)}{f_{m}^{(\alpha)}(x)},
\quad m\geq0.
\]
Note that
\[
\prod_{q=0}^{N-2}\langle P_q,Q_q\rangle
=\prod_{m=0}^{N-2}\Gamma(m+\alpha+2)
\frac{f_{m+1}^{(\alpha)}(x)}{f_{m}^{(\alpha)}(x)}
=\prod_{m=1}^{N-1}\Gamma(m+\alpha+1)
\cdot
\frac{f_{N-1}^{(\alpha)}(x)}{f^{(\alpha)}_0(x)},
\]
and hence
\[
\frac{N!}{Z_{N}}\prod_{q=0}^{N-2}\langle P_q,Q_q\rangle\cdot x^\alpha e^{-x}
=
\frac{N!}{ N!}\prod_{j=1}^N\frac{1}{\Gamma(j+\alpha)}\prod_{m=1}^{N-1}\Gamma(m+\alpha+1)\cdot \frac{f_{N-1}^{(\alpha)}(x)}{f_0^{(\alpha)}(x)}\cdot x^\alpha e^{-x} 
=
f_{N-1}^{(\alpha)}(x)x^\alpha e^{-x},
\]
where we used the fact $f_0^{(\alpha)}(x)=\frac{1}{\Gamma(\alpha+1)}$.
We recall that the inverse matrix of the lower triangular matrix $L$ is a lower triangular matrix and the inverse matrix of the upper triangular matrix $U$ is an upper triangular matrix. 
As a consequence, we see that 
\[
(L^{-1})_{pq}=
\begin{cases}
0 & q>p,\\
1 & q=p,\\
(\overline{\lambda})^{p-q}\frac{f^{(\alpha)}_{q}(x)}{f^{(\alpha)}_{p}(x)} & q<p,
\end{cases},
\quad
(U^{-1})_{pq}=
\begin{cases}
(\lambda)^{q-p}\frac{f^{(\alpha)}_{p}(x)}{f^{(\alpha)}_{q}(x)} & q>p,\\
1 & q=p,\\
0 & q<p.
\end{cases}
\]
Finally, we put
\[
G_{N-1}^{(\alpha)}(x|y,z)=\sum_{n,m=0}^{N-1}f_n^{(\alpha)}(x)f_m^{(\alpha)}(x)y^nz^m\sum_{k=\max\{n,m\}}^{N-1}
\frac{1}{\Gamma(k+\alpha+2)}\frac{x^k}{f_k^{(\alpha)}(x)f_{k+1}^{(\alpha)}(x)},
\]
then this completes the proof of Lemma~\ref{ConsOP}.
\end{proof}
\begin{rem}
If $\alpha=0$, then
\begin{equation*}
r_p=\Gamma(p+1)f_p^{(0)}(x)=p!\sum_{k=0}^p\frac{p+1-k}{k!}x^k=(p+1)!e_p(x)-p!xe_{p-1}(x),
\end{equation*}
where $e_p(x)=\sum_{k=0}^p\frac{x^k}{k!}$. This is consistent with the Ginibre unitary ensemble case \cite{ATTZ}.
\end{rem}
\begin{rem}
Our proof gives us the conditional expectation of the on-diagonal overlap. Indeed, we have
\[
\E\left[\cO_{1,1}|z_1=\lambda\right]
=
\frac{x-\alpha}{x}\frac{(N+\alpha+1)\left(e_N^{(\alpha)}(x)+\frac{\alpha}{\Gamma(\alpha+1)}\frac{1}{x-\alpha}\right)-x\left(e_{N-1}^{(\alpha)}(x)+\frac{\alpha}{\Gamma(\alpha+1)}\frac{1}{x-\alpha}\right)}{e_{N-1}^{(\alpha)}(x)}
\]
for $x=|\lambda|^2$. 
By the change of the variable $\lambda\mapsto \sigma\lambda$ and choosing $\sigma^2=N$ and $\alpha=Nb$, we have that for $\lambda\in\mathrm{int}S_{(\mathrm{reg})}$, 
\begin{equation}
\E\left[\cO_{1,1}|z_1=\lambda\right]=\frac{N(|\lambda|^2-b)\left(1+b-|\lambda|^2\right)}{|\lambda|^2}\left(1+o(1)\right),
\quad\text{as $N\to\infty$}. 
\label{aoE}
\end{equation}
This means that the conditional expectation of the on-diagonal overlap is affected by the inner boundary, the outer boundary, and a point insertion at the origin.
We mention that \eqref{aoE} already appeared in the physical paper \cite{Nowak_2018} by the diagrammatic approach, and hence, we showed the alternative proof here. 
We clearly obtain that (formally) as $b\to0$, 
\[
\E\left[\cO_{1,1}|z_1=\lambda\right]=N\left(1-|\lambda|^2\right)\left(1+o(1)\right),
\quad\text{as $N\to\infty$}, 
\]
which is nothing else but the conditional expectation of the on-diagonal overlap for Ginibre unitary ensemble, see \cite{Bourgade_Dubach:2021}.
Moreover, \eqref{aoE} tells us the order of the scaling. 
Indeed, for $p\in\mathrm{int}(S_{\mathrm{reg}})$ and \eqref{Zscaling}, $\E[\cO_{1,1}|z_1=\lambda]=O(N)$. On the other hand, if we consider the edge regime, $\E[\cO_{1,1}|z_1=\lambda]=O(\sqrt{N})$. Moreover, in the almost circular regime, $\E[\cO_{1,1}|z_1=\lambda]=O(1)$. These would be useful when we study the distribution itself of the on-diagonal overlap. 
\end{rem}
\begin{rem}
From the proof of Lemma~\ref{ConsOP}, we find that the planar orthogonal polynomials $\{P_k(\cdot|\lambda,\overline{\lambda})\}_{k}$ with respect to the weight function \eqref{eq:weight} are given by
\begin{equation}
P_k(z|\lambda,\overline{\lambda})=\sum_{j=0}^{k}\lambda^{k-j}\frac{f_j^{(\alpha)}(\lambda\overline{\lambda})}{f_k^{(\alpha)}(\lambda\overline{\lambda})}z^{j}.
\label{OPweight}
\end{equation}
Then, we can confirm that \eqref{OPweight} satisfies the non-standard three-term recurrence in the sense of \cite[Remark 1.3]{SungsooLee}.
More precisely, we have 
\begin{equation}
zP_k(z|\lambda,\overline{\lambda})
=
P_{k+1}(z|\lambda,\overline{\lambda})+b_kP_k(z|\lambda,\overline{\lambda})+c_kzP_{k-1}(z|\lambda,\overline{\lambda}),
\label{NSthree}
\end{equation}
where 
\begin{equation}
b_k=\overline{\lambda}\frac{f_k^{(\alpha)}(\lambda\overline{\lambda})}{f_{k+1}^{(\alpha)}(\lambda\overline{\lambda})}
\quad
c_k=-\overline{\lambda}\frac{f_{k-1}^{(\alpha)}(\lambda\overline{\lambda})}{f_k^{(\alpha)}(\lambda\overline{\lambda})}. 
\end{equation}
\end{rem}
\subsection{Simplification step}\label{s2-2}
In this subsection, we simplify \eqref{eq:PolynomialKernel2} in the reasonable form to take a large $N$-limit, and we will complete the proof of Theorem~\ref{thm:Thm1}. 
To this end, firstly, note that \eqref{eq:PolynomialKernel2} can be rewritten as 
\[
G_{N-1}^{(\alpha)}(x|y,z)=\sum_{n,m=0}^{N-1}f_n^{(\alpha)}(x)f_m^{(\alpha)}(x)y^nz^m
\left(
\Phi_{N-1}^{(\alpha)}(x)-\Phi_{\max\{n,m\}-1}^{(\alpha)}(x)
\right),
\]
where
\begin{equation}
\Phi_n^{(\alpha)}(x):=
\sum_{k=0}^{n}\frac{1}{\Gamma(k+\alpha+2)}\frac{x^k}{f_k^{(\alpha)}(x)f_{k+1}^{(\alpha)}(x)},\quad
\Phi_{-1}^{(\alpha)}(x)=0. \label{phiA}
\end{equation}
\begin{lem}
\eqref{phiA} can be simplified as follows:
\begin{equation}
\Phi_n^{(\alpha)}(x)
=
\Gamma(\alpha+1)\frac{x-(\alpha+1)}{(x-\alpha)x}
+
\frac{1}{x^2f_{n+1}^{(\alpha)}(x)}
\frac{x(n+\alpha+2-x)}{x-\alpha}
.\label{eq:DeformPhi}
\end{equation} 
\end{lem}
\begin{proof}
For fixed $x$, we have $
\Phi_{n+1}^{(\alpha)}(x)=
\Phi_n^{(\alpha)}(x)
+
\frac{1}{\Gamma(n+\alpha+3)}\frac{x^{n+1}}{f_{n+1}^{(\alpha)}(x)f_{n+2}^{(\alpha)}(x)}$ with $\Phi_0(x)=\frac{1}{x+\alpha+2}$. The proof is done by induction. From $f_1^{(\alpha)}(x)=\frac{x+\alpha+2}{\Gamma(\alpha+2)}$, the initical condition is satisfied. Now, we assume that \eqref{eq:DeformPhi} holds for $n\in\N$. Then, we have
\begin{equation*}
\Phi_{n+1}^{(\alpha)}(x)=\Gamma(\alpha+1)\frac{x-(\alpha+1)}{(x-\alpha)x}
+
\frac{1}{x^2f_{n+1}^{(\alpha)}(x)f_{n+2}^{(\alpha)}(x)}
\left\{\frac{x(n+\alpha+2-x)f_{n+2}^{(\alpha)}(x)}{x-\alpha}+\frac{x^{n+3}}{\Gamma(n+\alpha+3)}\right\}.
\end{equation*}
Here, it is easy to see that 
\begin{equation*}
\frac{x(n+\alpha+2-x)f_{n+2}^{(\alpha)}(x)}{x-\alpha}+\frac{x^{n+3}}{\Gamma(n+\alpha+3)}
=
\frac{x(n+\alpha+3-x)}{x-\alpha}f_{n+1}^{(\alpha)}(x),
\end{equation*}
which completes the proof. 
\end{proof}
\begin{rem}
If $\alpha=0$, then we clearly have 
$
\Phi_n(x):=\Phi_n^{(0)}(x)=\frac{x-1}{x^2}+\frac{1}{x^2f_{n+1}(x)}(n+2-x),
$
which is consistent with \cite{ATTZ}.
\end{rem}
Combining \eqref{eq:DeformPhi} with \eqref{eq:PolynomialKernel2}, we have
\begin{align}
G_{N-1}^{(\alpha)}(x|y,z)
=&
\frac{1}{x^2f_N^{(\alpha)}(x)}\frac{x(N+\alpha+1-x)}{x-\alpha}
\mu_{N-1}^{(\alpha)}(x,y)\mu_{N-1}^{(\alpha)}(x,z)
+
\frac{1}{x^2}\frac{x\left(x-\alpha-1-\omega\partial_\omega\right)}{x-\alpha}
\left.\sum_{n=0}^{N-1}f_n^{(\alpha)}(x)\omega^n\right|_{\omega=yz}\nonumber\\
&
+
\frac{1}{x^2}\frac{x\left(x-\alpha-1-z\partial_z\right)}{x-\alpha}
\sum_{n=0}^{N-1}\sum_{m=n+1}^{N-1}f_n^{(\alpha)}(x)y^nz^m
+
\frac{1}{x^2}\frac{x\left(x-\alpha-1-y\partial_y\right)}{x-\alpha}
\sum_{m=0}^{N-1}\sum_{n=m+1}^{N-1}f_m^{(\alpha)}(x)z^my^n.\nonumber
\end{align}
where we defined
\begin{equation}
\mu_{n}^{(\alpha)}(x,y):=\sum_{k=0}^{n}f_k^{(\alpha)}(x)y^k,\quad n\in\N.\label{eq:PolynomialKernel4}
\end{equation}
\subsubsection{Final step : Complete the proof of Theorem \ref{thm:Thm1}}\label{s2-3}
Now, we are ready to prove Theorem \ref{thm:Thm1}.
\begin{proof}[Proof of Theorem \ref{thm:Thm1}]
By elementary but involved computations, we have the following identities:
\[
\sum_{n=0}^{N-1}\sum_{m=n+1}^{N-1}f_n^{(\alpha)}(x)y^nz^m
=
\frac{z}{1-z}\mu_{N-1}^{(\alpha)}(x,yz)-\frac{z^N}{1-z}\mu_{N-1}^{(\alpha)}(x,y),
\]
\begin{align}
\mu_{N-1}^{(\alpha)}(x,t) \label{eq:GN1}
=&\frac{x-\alpha}{x}
\left\{\frac{\mathfrak{e}_{N-1}^{(\alpha)}(x\omega|x)}{(1-t)^2}
-\frac{1}{1-t}\frac{(xt)^N}{\Gamma(N+\alpha)}
-\frac{t^N(N+\alpha+1-x-(N+\alpha-x)t)}{(1-t)^2}
\mathfrak{e}_{N-1}^{(\alpha)}(x|x)
\right\},
\\
\partial_t\mu_{N-1}^{(\alpha)}(x,t) \label{eq:GN2}
=&
\left(\frac{N}{t}+\frac{2}{1-t}\right)\mu_{N-1}^{(\alpha)}(x,t)
-\frac{x-\alpha}{x}\frac{N+\alpha-xt}{t(1-t)^2}
e_{N-1}^{(\alpha)}(xt)+\frac{x-\alpha}{x}\frac{t^N(N+\alpha-x)e_{N-1}^{(\alpha)}(x)}{(1-t)^2}\\
&
-\frac{x-\alpha}{x}\frac{1}{t(1-t)^2}
\left\{
\frac{(xt)^N(1-t)}{\Gamma(N+\alpha)}+\frac{\alpha}{\Gamma(\alpha+1)}\frac{(N+\alpha-x)(1-t^{N+1})}{x-\alpha}
\right\},\nonumber
\end{align}
and
\begin{align}
x^2G_{N-1}^{(\alpha)}(x|y,z)
=&
\frac{x}{x-\alpha}\left\{\frac{(N+\alpha+1-x)}{f_N^{(\alpha)}(x)}\mu_{N-1}^{(\alpha)}(x,y)\mu_{N-1}^{(\alpha)}(x,z)\right.
\label{eq:GN3}
\\
&
+(x-\alpha-1)
\left(
\frac{1-yz}{(1-y)(1-z)}\mu_{N-1}^{(\alpha)}(x,yz)-\frac{y^N}{1-y}\mu_{N-1}^{(\alpha)}(x,z)-\frac{z^N}{1-z}\mu_{N-1}^{(\alpha)}(x,y)
\right)
\nonumber
\\
&+
\frac{Nz^N-(N-1)z^{N+1}}{(1-z)^2}\mu_{N-1}^{(\alpha)}(x,y)
+
\frac{Ny^N-(N-1)y^{N+1}}{(1-y)^2}\mu_{N-1}^{(\alpha)}(x,z)
\nonumber
\\
&\left.-
\frac{1-yz}{(1-y)(1-z)}t\partial_{t}\mu_{N-1}^{(\alpha)}(x,t)|_{t=yz}
-
\left(\frac{y}{(1-y)^2}+\frac{z}{(1-z)^2}\right)\mu_{N-1}^{(\alpha)}(x,yz)\right\}.
\nonumber
\end{align}
Notice also that 
\begin{equation}
\left(\frac{Nz^N-(N-1)z^{N+1}}{(1-z)^2}-(x-\alpha-1)\frac{z^N}{1-z}\right)\mu_{N-1}^{(\alpha)}(x,y)
=
z^N\left(\frac{N+\alpha-x}{1-z}+\frac{1}{(1-z)^2}\right)\mu_{N-1}^{(\alpha)}(x,y).
\label{eq:GN4}
\end{equation}
Using \eqref{eq:GN1}, \eqref{eq:GN2}, \eqref{eq:GN3}, \eqref{eq:GN4},
and grouping by denominators, we can summarize as follows:
\[
x(x-\alpha)G_{N-1}^{(\alpha)}(x|y,z)
=
\frac{T_A^{(\alpha)}(x,y,z)}{(1-y)^2(1-z)^2}
+
\frac{T_B^{(\alpha)}(x,y,z)}{(1-y)(1-z)}
+
\frac{T_C^{(\alpha)}(x,y,z)}{(1-y)^2(1-z)}+\frac{T_C^{(\alpha)}(x,z,y)}{(1-y)(1-z)^2},
\]
where 
\begin{align*}
\left(\frac{x}{x-\alpha}\right)^2f_N^{(\alpha)}(x)T_A^{(\alpha)}(x,y,z)
=&(N+\alpha+1)
\left\{
\mathfrak{e}_{N}^{(\alpha)}(xy|x)
\mathfrak{e}_{N}^{(\alpha)}(xz|x)
-\mathfrak{e}_{N-1}^{(\alpha)}(xyz|x)\mathfrak{e}_{N-1}^{(\alpha)}(x|x)
\right\}\\
&-x\left\{
\mathfrak{e}_{N-1}^{(\alpha)}(xy|x)
\mathfrak{e}_{N-1}^{(\alpha)}(xz|x)
-\mathfrak{e}_{N-1}^{(\alpha)}(xyz|x)\mathfrak{e}_{N-1}^{(\alpha)}(x|x)
\right\},
\end{align*}
\begin{align*}
\left(\frac{x}{x-\alpha}\right)^2f_N^{(\alpha)}(x)T_B^{(\alpha)}(x,y,z)
=&
x\frac{(xy)^N(xz)^N}{\Gamma(N+\alpha+1)\Gamma(N+\alpha)}
+
x(N+\alpha-x)\mathfrak{e}_{N-1}^{(\alpha)}(x|x)
\frac{(xyz)^N}{\Gamma(N+\alpha+1)}\\
&+
xe_{N-1}^{(\alpha)}(xyz)\left\{
(N+\alpha+1-x)\mathfrak{e}_{N-1}^{(\alpha)}(x|x)
+
\frac{(N+\alpha+1)x^N}{\Gamma(N+\alpha+1)}
\right\},
\end{align*}
\begin{equation*}
\left(\frac{x}{x-\alpha}\right)^2f_N^{(\alpha)}(x)T_C^{(\alpha)}(x,y,z)
=
x\frac{(xyz)^N}{\Gamma(N+\alpha+1)}\mathfrak{e}_{N-1}^{(\alpha)}(x|x)
-x\frac{(xz)^N}{\Gamma(N+\alpha+1)}\mathfrak{e}_{N-1}^{(\alpha)}(xy|x).
\end{equation*}
Then, we obtain 
\begin{align*}
x^2\frac{x}{x-\alpha}f_N^{(\alpha)}(x)G_{N-1}^{(\alpha)}(x|y,z)
=&
\frac{(N+\alpha+1)W_{N}^{(\alpha)}(x,y,z)-xW_{N-1}^{(\alpha)}(x,y,z)}{(1-y)^2(1-z)^2}
+x
\frac{(xyz)^N\mathfrak{e}_{N-1}^{(\alpha)}(x|x)
-(xz)^N\mathfrak{e}_{N-1}^{(\alpha)}(xy|x)}{\Gamma(N+\alpha+1)(1-y)^2(1-z)}\\
&
+x
\frac{(xyz)^N\mathfrak{e}_{N-1}^{(\alpha)}(x|x)
-(xy)^N\mathfrak{e}_{N-1}^{(\alpha)}(xz|x)}{\Gamma(N+\alpha+1)(1-y)(1-z)^2}
-x\frac{(xyz)^N}{\Gamma(N+\alpha+1)}\frac{\mathfrak{e}_{N}^{(\alpha)}(x|x)+x\mathfrak{e}_{N-1}^{(\alpha)}(x|x)}{(1-y)(1-z)}
\\
&
-\frac{1}{\Gamma(\alpha)}\frac{x}{x-\alpha}\frac{
(N+\alpha+1)\mathfrak{e}_{N}^{(\alpha)}(x|x)-x\mathfrak{e}_{N-1}^{(\alpha)}(x|x)}{(1-y)(1-z)},
\end{align*}
where we set
\[
W_N^{(\alpha)}(x,y,z)
:=
\mathfrak{e}_{N}^{(\alpha)}(xy|x)
\mathfrak{e}_{N}^{(\alpha)}(xz|x)
-(1-x(1-y)(1-z))
\mathfrak{e}_{N}^{(\alpha)}(xyz|x)
\mathfrak{e}_{N}^{(\alpha)}(x|x).
\]
Using $e_{n-1}^{(\alpha)}(x)=e_{n}^{(\alpha)}(x)-\frac{x^n}{\Gamma(n+\alpha+1)}$, we see that 
\begin{align*}
W_N^{(\alpha)}(x,y,z)
=&W_{N+1}^{(\alpha)}(x,y,z)
-\frac{(xz)^{N+1}\mathfrak{e}_{N+1}^{(\alpha)}(xy|x)}{\Gamma(N+\alpha+2)}
-\frac{(xy)^{N+1}\mathfrak{e}_{N+1}^{(\alpha)}(xz|x)}{\Gamma(N+\alpha+2)}
+\frac{x(1-y)(1-z)(x^2yz)^{N+1}}{\Gamma(N+\alpha+2)^2}\\
&+(1-x(1-y)(1-z))
\left\{
\frac{x^{N+1}\mathfrak{e}_{N+1}^{(\alpha)}(xyz|x)}{\Gamma(N+\alpha+2)}
+
\frac{(xyz)^{N+1}\mathfrak{e}_{N+1}^{(\alpha)}(x|x)}{\Gamma(N+\alpha+2)}
\right\}.
\end{align*}
Therefore, we can find that 
\begin{align*}
x^2\frac{x}{x-\alpha}f_N^{(\alpha)}(x)G_{N-1}^{(\alpha)}(x|y,z)
=&
\frac{(N+\alpha+1)W_{N+1}^{(\alpha)}(x,y,z)-xW_{N}^{(\alpha)}(x,y,z)}{(1-y)^2(1-z)^2}
-\frac{x}{(1-y)(1-z)}\frac{(xyz)^{N+1}\mathfrak{e}_{N+1}^{(\alpha)}(x|x)}{\Gamma(N+\alpha+1)}\\
&
-\frac{1}{\Gamma(\alpha)}\frac{x}{x-\alpha}\frac{
(N+\alpha+1)\mathfrak{e}_{N}^{(\alpha)}(x|x)-x\mathfrak{e}_{N-1}^{(\alpha)}(x|x)}{(1-y)(1-z)}.
\end{align*}
Here, the third term can be written as
\[
-\frac{1}{\Gamma(\alpha)}\frac{x}{x-\alpha}
\frac{(N+\alpha+1)\mathfrak{e}_{N}^{(\alpha)}(x|x)-x\mathfrak{e}_{N-1}^{(\alpha)}(x|x)}{(1-y)(1-z)}
=-\frac{1}{\Gamma(\alpha)}\left(\frac{x}{x-\alpha}\right)^2\frac{f_N^{(\alpha)}(x)}{(1-y)(1-z)}.
\]
Hence, we have
\begin{align}
G_{N-1}^{(\alpha)}(x|y,z)
=&
\frac{(N+\alpha+1)W_{N+1}^{(\alpha)}(x,y,z)-xW_{N}^{(\alpha)}(x,y,z)}{x^2(1-y)^2(1-z)^2g_N^{(\alpha)}(x)}
\label{eq:GNF1}\\
&
-\frac{1}{x(1-y)(1-z)g_N^{(\alpha)}(x)}\frac{(xyz)^{N+1}\mathfrak{e}_{N+1}^{(\alpha)}(x|x)}{\Gamma(N+\alpha+1)}
-\frac{1}{\Gamma(\alpha)}\frac{1}{x(x-\alpha)}\frac{1}{(1-y)(1-z)}
\nonumber
\end{align}
Letting $x=\lambda\overline{\lambda},y=\frac{\overline{z}}{\overline{\lambda}},z=\frac{w}{\lambda}$, then we have
\begin{align*}
G_{N-1}^{(\alpha)}\left(\lambda\overline{\lambda}\Bigr|\frac{\overline{z}}{\overline{\lambda}},\frac{w}{\lambda}\right)
=&
\frac{\lambda\overline{\lambda}-\alpha}{\lambda\overline{\lambda}}
\frac{(N+\alpha+1)\mathrm{I}_{N+1}^{(\alpha)}(\overline{z},w|\lambda,\overline{\lambda})-
\lambda\overline{\lambda}\mathrm{I}_{N}^{(\alpha)}(\overline{z},w|\lambda,\overline{\lambda})}{(\overline{z}-\overline{\lambda})^2(w-\lambda)^2f_N(\lambda\overline{\lambda})
\varpi(\overline{z},w)\varpi(\overline{\lambda},\lambda)}\\
&
-\frac{\lambda\overline{\lambda}-\alpha}{\lambda\overline{\lambda}}
\frac{(\overline{z}w)^{N+1}\mathfrak{e}_{N+1}^{(\alpha)}(\lambda\overline{\lambda}|\lambda\overline{\lambda})}{\Gamma(N+\alpha+1)(\overline{z}-\overline{\lambda})(w-\lambda)f_N^{(\alpha)}(\lambda\overline{\lambda})}
-\frac{1}{\Gamma(\alpha)}\frac{1}{(\lambda\overline{\lambda}-\alpha)(\overline{z}-\overline{\lambda})(w-\lambda)}\\
=&
\mathfrak{I}_N^{(\alpha)}(\overline{z},w|\overline{\lambda},\lambda)
+
\mathfrak{II}_N^{(\alpha)}(\overline{z},w|\overline{\lambda},\lambda)
+
\mathfrak{III}_N^{(\alpha)}(\overline{z},w|\overline{\lambda},\lambda).
\end{align*}
This completes the proof. 
\end{proof}
\section{Scaling limits: Proof of Theorem}\label{s5}
In below all proofs, a constant $\epsilon>0$ is independent of $N$, but it may be different in each line. 
It does not affect the our desired result. 
Also, note that the conjugation of kernel the $K(x,y)$, $K(x,y)\mapsto\phi(x)K(x,y)\phi^{-1}(y)$, does not change the value of the determinant.
Hence, in order to simplify displays in the asymptotic expansion in the below proofs, 
we often omit to write such conjugation factor. 
That is, we only write the terms contributed to the limiting kernels in the below proofs, and we do not mention that in each line. 
\subsection{Proof of the bulk case in the strongly non-unitary regime in Theorem~\ref{MainThmS}}
We fix $a_N=1$ and $\alpha=b_N=Nb$. 
For $p\in \mathrm{int}(S_{(\mathrm{reg})})$,
let $z=\sqrt{N}p+\zeta,w=\sqrt{N}p+\eta,\lambda=\sqrt{N}p+\chi$.
It is straightforward to see 
\begin{equation}
\mathfrak{e}_N^{(Nb)}(\overline{z}w|\lambda\overline{\lambda})\varpi(\overline{z},w)
=e^{-\frac{1}{2}(|\zeta|^2+|\eta|^2)+\eta\overline{\zeta}}+O\left(e^{-\epsilon N}\right),
\label{Ab1}
\end{equation}
uniformly for $\zeta,\eta,\chi$ in compact subsets of $\C$. 
Therefore, we obtain
\begin{equation}
f_{N}^{(Nb)}(\lambda\overline{\lambda})\varpi(\lambda,\overline{\lambda})
=
N\frac{(|p|^2-b)(1+b-|p|^2)}{|p|^2}\left(1+O\left(e^{-\epsilon N}\right)\right),
\label{Ab2}
\end{equation}
uniformly for $\chi$ in a compact subset of $\C$. 
With help of these asymptotic expansions, we shall look at the asymptotic expansions of \eqref{KIA}, \eqref{KII}, and \eqref{KIII}.
First, by the simple and long computations, we have  
\begin{equation}
\mathfrak{I}_{N}^{(Nb)}(\overline{z},w|\lambda,\overline{\lambda})\omega(\overline{z},z|\lambda,\overline{\lambda})
=
\left. \frac{d}{du}\left(\frac{e^{u}-1}{u}\right)\right|_{u=(\overline{\zeta}-\overline{\chi})(\eta-\chi)}(1+(\overline{\zeta}-\overline{\chi})(\zeta-\chi))e^{-(\overline{\zeta}-\overline{\chi})(\zeta-\chi)}
(1+O(N^{-1})),
\label{MTB1}
\end{equation}
uniformly for $\zeta,\eta,\chi$ in compact subsets of $\C$.
Here, it is easy to see that 
\begin{equation}
\mathfrak{II}_N^{(Nb)}(\overline{z},w|\lambda,\overline{\lambda})
\omega(\overline{z},z|\lambda,\overline{\lambda})
=
O\left(e^{-\epsilon N}\right)
\label{MTB2}
\end{equation}
and
\begin{equation}
\mathfrak{III}_N^{(Nb)}(\overline{z},w|\lambda,\overline{\lambda})\omega(\overline{z},z|\lambda,\overline{\lambda})
=
O\left(e^{-\epsilon N}\right),
\label{MTB3}
\end{equation}
uniformly for $\zeta,\eta,\chi$ in compact subsets of $\C$ since the main contributions in \eqref{KII} and \eqref{KIII} are come from the asymptotics of the Gamma function. 
Combining \eqref{MTB1} with \eqref{MTB2} and \eqref{MTB3}, we have 
\begin{equation}
\lim_{N\to\infty}K_{1,1}^{(N)}(z,\zbar,w,\wbar|\lambda,\overline{\lambda}):=
\cK_{1,1}^{(\mathrm{Bulk})}(\overline{\zeta},\eta|\chi,\overline{\chi})
\omega^{(\mathrm{Bulk})}(\zeta,\overline{\zeta}|\chi,\overline{\chi}),
\end{equation}
where the convergence is uniform for $\zeta,\eta,\chi$ in compact subsets of $\C$, and we defined
\[
\omega^{(\mathrm{Bulk})}(\zeta,\overline{\zeta}|\chi,\overline{\chi})=
\left(1+(\overline{\zeta}-\overline{\chi})(\zeta-\chi)\right)e^{-(\overline{\zeta}-\overline{\chi})(\zeta-\chi)},
\quad
\cK_{1,1}^{(\mathrm{Bulk})}(\overline{\zeta},\eta|\chi,\overline{\chi})=
\left.\frac{d}{du}\left(\frac{e^u-1}{u}\right)\right|_{u=(\overline{\zeta}-\overline{\chi})(\eta-\chi)}. 
\]
For the proof of the off-diagonal case, since we can easily calculate the scaling limit similar to \cite{ATTZ}, we omit the details here.
\subsection{Proof of the edge case in the strongly non-unitary regime in Theorem~\ref{MainThmS}}
We fix $a_N=1$ and $\alpha=b_N=Nb$ again. 
We use the following uniform asymptotic behaviour of the incomplete Gamma function \cite[equation (8.8.9)]{OLB}:
\begin{equation}
Q(s+1,s+\sqrt{s}z)=\frac{1}{2}\mathrm{erfc}\left(\frac{z}{\sqrt{2}}\right)+\frac{e^{-\frac{z^2}{2}}}{\sqrt{2\pi s}}
\frac{2+z^2}{3}+O\left(\frac{1}{s}\right)\quad \text{as $s\to\infty$}, 
\label{eq:AsymQ}
\end{equation}
uniformly for $z$ in a compact subset of $\C$. 
Here, we recall that   
\begin{equation}
\mathrm{erfc}(z)=\frac{e^{-z^2}}{\sqrt{\pi}z}\left(1+O(z^{-2})\right)\quad\text{as $z\to\infty$ with $|\mathrm{arg}(z)|<\frac{3}{4}\pi$},
\end{equation}
and by the relation $\mathrm{erfc}(-z)=2-\mathrm{erfc}(z)$, 
\begin{equation}
\mathrm{erfc}(-z)=2-\frac{e^{-z^2}}{\sqrt{\pi}z}\left(1+O(z^{-2})\right)\quad\text{as $z\to\infty$ with $|\mathrm{arg}(z)|<\frac{3}{4}\pi$}.
\end{equation}
For the outer edge case, we set$z=e^{i\theta}(\sqrt{N(b+1)}+\zeta),w=e^{i\theta}(\sqrt{N(b+1)}+\eta),\lambda=e^{i\theta}(\sqrt{N(b+1)}+\chi)$ as in \eqref{Zscaling}. 
For the inner edge case, we set 
$z=e^{i\theta}(\sqrt{Nb}-\zeta),w=e^{i\theta}(\sqrt{Nb}-\eta),\lambda=e^{i\theta}(\sqrt{Nb}-\chi)$ as in \eqref{Zscaling}. 
For the outer edge case, using \eqref{eq:AsymQ}, we obtain 
\begin{equation}
\label{AE1O}
\mathfrak{e}_{N+k-1}^{(Nb)}(z\overline{w}|\lambda\overline{\lambda})
\varpi(z,\overline{w})
=
e^{-\frac{1}{2}(|\zeta|^2+|\eta|^2)+\zeta\overline{\eta}}
\Bigl(F(\zeta+\overline{\eta})+\frac{e^{-\frac{(\zeta+\overline{\eta})^2}{2}}}{\sqrt{2\pi N(1+b)}}\bigl(\frac{(\zeta+\overline{\eta})^2+2}{3}+k\bigr)+O(N^{-1})\Bigr)
\end{equation}
and for the inner edge case, similar to the above, we have 
\begin{equation}
\label{AE1I}
\mathfrak{e}_{N+k-1}^{(Nb)}(z\overline{w}|\lambda\overline{\lambda})
\varpi(z,\overline{w})
=
e^{-\frac{1}{2}(|\zeta|^2+|\eta|^2)+\zeta\overline{\eta}}
\Bigl(F(\zeta+\overline{\eta})-\frac{e^{-\frac{(\zeta+\overline{\eta})^2}{2}}}{\sqrt{2\pi}(\chi+\overline{\chi})} +O(N^{-1/2})\Bigr),
\end{equation}
uniformly for $\zeta,\eta,\chi$ in compact subsets of $\C$.
Here, $F(u)$ is the complementary error function defined in \eqref{errorF}.
Using \eqref{AE1O} and \eqref{AE1I}, we can find that 
for the outer edge case,
\begin{equation}
\label{AE2O}
f_{N}^{(Nb)}(\lambda\overline{\lambda})\varpi(\lambda,\overline{\lambda})
=
\sqrt{\frac{N}{1+b}}\cF(\chi+\overline{\chi})
\left(1+O(N^{-1/2})\right)
\end{equation}
and for the inner edge case,
\begin{equation}
\label{AE2I}
f_{N}^{(Nb)}(\lambda\overline{\lambda})\varpi(\lambda,\overline{\lambda})
=
\sqrt{\frac{N}{b}}\cF(\chi+\overline{\chi})(1+O(N^{-1/2})),
\end{equation}
uniformly for $\zeta,\eta,\chi$ in compact subsets of $\C$.
We now look at the asymptotics of \eqref{KIA}, \eqref{KII}, and \eqref{KIII} using \eqref{AE1O} and \eqref{AE1I}. 
For \eqref{KIA}, observe that for the outer edge case, 
\begin{align}
\label{AE1O}
\mathfrak{I}_{N}^{(bN)}(\overline{z},w|\overline{\lambda},\lambda)
\omega(\overline{z},z|\overline{\lambda},\lambda)
=&
\Biggl\{
e^{\overline{\zeta}\eta}\frac{H(\overline{\chi}+\chi,\overline{\zeta}+\chi,\overline{\chi}+\eta,\overline{\zeta}+\eta,(\overline{\zeta}-\overline{\chi})(\eta-\chi))}{(\overline{\zeta}-\overline{\chi})^2(\eta-\chi)^2}
\\
&
\quad\quad
+
e^{\overline{\zeta}\eta}\frac{\sqrt{2\pi}e^{-\frac{1}{2}(\overline{\zeta}+\eta)^2}F(\chi+\overline{\chi})}{(\overline{\zeta}-\overline{\chi})(\eta-\chi)\cF(\overline{\chi}+\chi)}
\Biggr\}
\omega^{(\mathrm{edge})}(\overline{\zeta},\zeta|\overline{\chi},\chi)
(1+O(N^{-1/2})),
\nonumber
\end{align}
and for the inner edge case, 
\begin{align}
\label{AE1I}
\mathfrak{I}_{N}^{(bN)}(\overline{z},w|\overline{\lambda},\lambda)
\omega(\overline{z},z|\overline{\lambda},\lambda)
=&
\Biggl\{
e^{\overline{\zeta}\eta}
\frac{H(\chi+\overline{\chi},\chi+\overline{\zeta},\overline{\chi}+\eta,\overline{\zeta}+\eta,(\overline{\zeta}-\overline{\chi})(\eta-\chi))}{(\overline{\zeta}-\overline{\chi})^2(\eta-\chi)^2}
\omega^{(\mathrm{edge})}(\overline{\zeta},\zeta|\overline{\chi},\chi)
\\
&\quad\quad
+
\frac{1}{\sqrt{2\pi}(\overline{\chi}+\chi)}
\frac{(1+(\overline{\zeta}-\overline{\chi})(\zeta-\chi))e^{-\frac{1}{2}(\zeta+\overline{\zeta})^2}}{(\overline{\zeta}-\overline{\chi})(\eta-\chi)}
\Biggr\}
(1+O(N^{-1/2})),
\nonumber
\end{align}
uniformly for $\zeta,\eta,\chi$ in compact subsets of $\C$.
For \eqref{KII}, observe that for the outer edge case, 
\begin{equation}
\label{AE2O}
\mathfrak{II}_{N}^{(bN)}(\overline{z},w|\overline{\lambda},\lambda)
\omega(\overline{z},z|\overline{\lambda},\lambda)
=
-
e^{\overline{\zeta}\eta}\frac{\sqrt{2\pi}e^{-\frac{1}{2}(\overline{\zeta}+\eta)^2}F(\chi+\overline{\chi})}
{(\overline{\zeta}-\overline{\chi})(\eta-\chi)\cF(\overline{\chi}+\chi)}\omega^{(\mathrm{edge})}(\overline{\zeta},\zeta|\overline{\chi},\chi)
(1+O(N^{-1/2})),
\end{equation}
and for the inner edge case, 
\begin{equation}
\label{AE2I}
\mathfrak{II}_{N}^{(bN)}(\overline{z},w|\overline{\lambda},\lambda)
\omega(\overline{z},z|\overline{\lambda},\lambda)
=
O(e^{-\epsilon N})
\end{equation}
uniformly for $\zeta,\eta,\chi$ in compact subsets of $\C$.
For \eqref{KIII}, observe that for the outer edge case, 
\begin{equation}
\label{AE3O}
\mathfrak{III}_{N}^{(bN)}(\overline{z},w|\overline{\lambda},\lambda)
\omega(\overline{z},z|\overline{\lambda},\lambda)
=
O\left(e^{-\epsilon N}\right) 
\end{equation}
and for the inner edge case, 
\begin{equation}
\label{AE3I}
\mathfrak{III}_{N}^{(bN)}(\overline{z},w|\overline{\lambda},\lambda)
\omega(\overline{z},z|\overline{\lambda},\lambda)
=
-
\frac{1}{\sqrt{2\pi}(\chi+\overline{\chi})}\frac{(1+(\overline{\zeta}-\overline{\chi})(\zeta-\chi))e^{-\frac{1}{2}(\zeta+\overline{\zeta})^2}}{(\overline{\zeta}-\overline{\chi})(\eta-\chi)}(1+O(N^{-1/2})),
\end{equation}
uniformly for $\zeta,\eta,\chi$ in compact subsets of $\C$.
Combining the asymptotics of \eqref{AE1O}, \eqref{AE2O}, and \eqref{AE3O} in the  outer edge case and \eqref{AE1I}, \eqref{AE2I}, and \eqref{AE3I} in the inner edge case, we obtain 
\begin{align*}
K_{11}^{(\mathrm{Edge})}\left(\zeta,\overline{\zeta},\eta,\overline{\eta}|\chi,\overline{\chi}\right)
:=&\lim_{N\to\infty}K_{11}^{(N)}(z,\zbar,w,w|\lambda,\overline{\lambda})\\
=&
\left(1+(\overline{\zeta}-\overline{\chi})(\zeta-\chi)\right)e^{-\zeta\overline{\zeta}}
e^{\overline{\zeta}\eta}
\frac{H(\chi+\overline{\chi},\chi+\overline{\zeta},\overline{\chi}+\eta,\overline{\zeta}+\eta,(\overline{\zeta}-\overline{\chi})(\eta-\chi))}{(\overline{\zeta}-\overline{\chi})^2(\eta-\chi)^2},
\nonumber
\end{align*}
uniformly for $\zeta,\eta,\chi$ in compact subsets of $\C$.
For the off-diagonal case, by applying Lemma \ref{thm:lem1} and just performing the similar computations, we get the desired result. 
\subsection{Proof of the weakly non-unitary regime in Theorem~\ref{MainThmS}}
We fix $a_N=\frac{N}{\rho^2}$ and $\alpha=b_N=N\left(\frac{N}{\rho^2}-\frac{1}{2}\right)$.
As in the strong non-unitary regime, we use \eqref{eq:AsymQ} many times. 
For $\theta\in[0,2\pi)$ and $\chi,\zeta,\eta$ in a compact subset of complex plane $\C$, let $z=e^{i\theta}(\sqrt{Na_N}+\zeta),w=e^{i\theta}(\sqrt{Na_N}+\eta),\lambda=e^{i\theta}(\sqrt{Na_N}+\chi)$ as in \eqref{Zscaling}. 
Let 
\begin{equation}
\cJ_{\rho}(\zeta,\overline{\eta}|\chi,\overline{\chi})
=
L_{\rho}(\zeta+\overline{\eta})+\frac{e^{-\frac{1}{2}\left(\zeta+\overline{\eta}+\frac{\rho}{2}\right)^2}}{\sqrt{2\pi}\left(\chi+\overline{\chi}+\frac{\rho}{2}\right)}.
\label{cJ}
\end{equation}
By \eqref{eq:AsymQ}, we have
\begin{equation}
\label{AW1}
\mathfrak{e}_{N+k-1}^{(b_N)}(z\wbar|\lambda\overline{\lambda})
\varpi(z,\overline{w})
=
e^{-\frac{1}{2}(|z|^2+|w|^2)+z\overline{w}}
\Bigl\{
\cJ_{\rho}(\zeta,\overline{\eta}|\chi,\overline{\chi})
+
\frac{\rho}{N}\Bigl(k\frac{e^{-\frac{1}{2}\left(\zeta+\overline{\eta}-\frac{\rho}{2}\right)^2}}{\sqrt{2\pi}}+C(\zeta,\overline{\eta})\Bigr)
\Bigr\}
(1+O(N^{-2})),
\end{equation}
uniformly for $\zeta,\eta,\chi$ in compact subsets of $\C$. 
Here, we defined
\[
C(\zeta,\overline{\eta})
=
\frac{1}{\sqrt{2\pi}}
\Biggl\{
\frac{2+\left(\zeta+\overline{\eta}-\frac{\rho}{2}\right)^2}{3}e^{-\frac{1}{2}\left(\zeta+\overline{\eta}-\frac{\rho}{2}\right)^2}
-
\frac{2+\left(\zeta+\overline{\eta}+\frac{\rho}{2}\right)^2}{3}e^{-\frac{1}{2}\left(\zeta+\overline{\eta}+\frac{\rho}{2}\right)^2}
\Biggr\},
\]
which does not depend on $k$. Here, the constant $C(\zeta,\overline{\eta})$ does not play the central role as we will see later. 
From \eqref{AW1}, we see that 
\begin{equation}
\label{AW2}
f_N^{(b_N)}(\lambda\overline{\lambda})
\varpi(\lambda,\overline{\lambda})
=
\cL_{\rho}(\chi+\overline{\chi})
(1+O(N^{-1})),
\end{equation}
uniformly for $\chi$ in a compact subset of $\C$. 
Now, we shall look at the asymptotics of \eqref{KIA}, \eqref{KII}, and \eqref{KIII}. 
Using asymptotics \eqref{AW1} and \eqref{AW2}, we shall firstly see the asymptotics of \eqref{KII} and \eqref{KIII}. 
Observe that 
\[
\mathfrak{II}_N^{(b_N)}(\overline{z},w|\lambda,\overline{\lambda})
\omega(\overline{z},z|\lambda,\overline{\lambda})
=
-\frac{\bigl(\chi+\overline{\chi}+\frac{\rho}{2}\bigr)}{\sqrt{2\pi}}
\frac{(1+|\zeta-\chi|^2)\cJ_{\rho}(\chi,\overline{\chi}|\chi,\overline{\chi})}{\cL_{\rho}(\chi+\overline{\chi})(\overline{\zeta}-\overline{\chi})(\eta-\chi)}
e^{\frac{1}{2}(\overline{\zeta}+\eta+\frac{\rho}{2})^2-\frac{1}{2}(\overline{\zeta}+\eta-\frac{\rho}{2})^2-\frac{1}{2}(\zeta+\overline{\zeta}+\frac{\rho}{2})^2}
(1+O(N^{-1})),
\]
uniformly for $\zeta,\eta,\chi$ in compact subsets of $\C$.
We also observe that 
\[
\mathfrak{III}_N^{(b_N)}(\overline{z},w|\lambda,\overline{\lambda})
\omega(\overline{z},z|\lambda,\overline{\lambda})
=
-\frac{(1+|\zeta-\chi|^2)e^{-\frac{1}{2}(\zeta+\overline{\zeta}+\frac{\rho}{2})^2}}{\sqrt{2\pi}(\chi+\overline{\chi}+\frac{\rho}{2})(\overline{\zeta}-\overline{\chi})(\eta-\chi)}(1+O(N^{-1})),
\]
uniformly for $\zeta,\eta,\chi$ in compact subsets of $\C$.
For the later purpose, it is convenient to write 
\begin{align}
\label{WII}
\mathfrak{II}_N^{(b_N)}(\overline{z},w|\lambda,\overline{\lambda})
\omega(\overline{z},z|\lambda,\overline{\lambda})
=&
-\Bigl\{
\frac{(1+|\zeta-\chi|^2)e^{\frac{1}{2}(\overline{\zeta}+\eta+\frac{\rho}{2})^2-\frac{1}{2}(\overline{\zeta}+\eta-\frac{\rho}{2})^2}e^{-\frac{1}{2}(\zeta+\overline{\zeta}+\frac{\rho}{2})^2}}{\sqrt{2\pi}\cL_{\rho}(\chi+\overline{\chi})(\overline{\zeta}-\overline{\chi})(\eta-\chi)}
\bigl(\chi+\overline{\chi}+\frac{\rho}{2}\bigr)L_{\rho}(\chi+\overline{\chi})
\\
&
+
\frac{(1+|\zeta-\chi|^2)e^{\frac{1}{2}(\overline{\zeta}+\eta+\frac{\rho}{2})^2-\frac{1}{2}(\overline{\zeta}+\eta-\frac{\rho}{2})^2}e^{-\frac{1}{2}(\zeta+\overline{\zeta}+\frac{\rho}{2})^2}e^{-\frac{1}{2}(\chi+\overline{\chi}+\frac{\rho}{2})^2}}{\sqrt{2\pi}^2\cL_{\rho}(\chi+\overline{\chi})(\overline{\zeta}-\overline{\chi})(\eta-\chi)}
\Bigr\}
(1+O(N^{-1})),
\nonumber
\end{align}
and 
\begin{align}
\label{WIII}
\mathfrak{III}_N^{(b_N)}(\overline{z},w|\lambda,\overline{\lambda})
\omega(\overline{z},z|\lambda,\overline{\lambda})
=&
-\Bigl\{
\frac{(1+|\zeta-\chi|^2)e^{-\frac{1}{2}(\zeta+\overline{\zeta}+\frac{\rho}{2})^2-\frac{1}{2}(\chi+\overline{\chi}-\frac{\rho}{2})^2}}{\sqrt{2\pi}^2\cL_{\rho}(\chi+\overline{\chi})(\overline{\zeta}-\overline{\chi})(\eta-\chi)}
\\
&
-\frac{(1+|\zeta-\chi|^2)e^{-\frac{1}{2}(\zeta+\overline{\zeta}+\frac{\rho}{2})^2-\frac{1}{2}(\chi+\overline{\chi}+\frac{\rho}{2})^2}(\chi+\overline{\chi}-\frac{\rho}{2})}{\sqrt{2\pi}^2\cL_{\rho}(\chi+\overline{\chi})(\overline{\zeta}-\overline{\chi})(\eta-\chi)(\chi+\overline{\chi}+\frac{\rho}{2})}
\nonumber
\\
&
-\frac{(\chi+\overline{\chi}-\frac{\rho}{2})(1+|\zeta-\chi|^2)e^{-\frac{1}{2}(\zeta+\overline{\zeta}+\frac{\rho}{2})^2}L_{\rho}(\chi+\overline{\chi})}{\sqrt{2\pi}\cL_{\rho}(\chi+\overline{\chi})(\overline{\zeta}-\overline{\chi})(\eta-\chi)}
\Bigr\}(1+O(N^{-1})),
\nonumber
\end{align}
where we used \eqref{cJ}.
Next, we shall look at the asymptotics of \eqref{KIA}. 
First, note that 
\begin{align}
\label{FocusEq}
(N+b_N+1)\mathrm{I}_{N+1}^{(b_N)}(\overline{z},w|\lambda,\overline{\lambda})
-
\lambda\overline{\lambda}\mathrm{I}_{N}^{(b_N)}(\overline{z},w|\lambda,\overline{\lambda})
=&
\frac{N^2}{\rho^2}
(\mathrm{I}_{N+1}^{(b_N)}(\overline{z},w|\lambda,\overline{\lambda})
-\mathrm{I}_{N}^{(b_N)}(\overline{z},w|\lambda,\overline{\lambda}))
\\
&
+\Bigl(\frac{N}{2}\mathrm{I}_{N+1}^{(b_N)}(\overline{z},w|\lambda,\overline{\lambda})-\frac{N(\chi+\overline{\chi})}{\rho}\mathrm{I}_{N}^{(b_N)}(\overline{z},w|\lambda,\overline{\lambda})\Bigr)
\nonumber
\\
&+\mathrm{I}_{N+1}^{(b_N)}(\overline{z},w|\lambda,\overline{\lambda})-|\chi|^2\mathrm{I}_{N}^{(b_N)}(\overline{z},w|\lambda,\overline{\lambda}).
\nonumber
\end{align}
In order to see the asymptotic behavior of \eqref{FocusEq}, we observe that for $k=0,1$, 
\begin{align*}
\mathrm{I}_{N+k}^{(b_N)}(\overline{z},w|\lambda,\overline{\lambda})
=&
e^{-\frac{1}{2}(|z|^2+|w|^2)+w\overline{z}}
\bigl(
e^{-(\overline{\zeta}-\overline{\chi})(\eta-\chi)}\mathcal{I}_{N,k}^{(1)}
-(1-(\overline{\zeta}-\overline{\chi})(\eta-\chi))\mathcal{I}_{N,k}^{(2)}
\bigr),
\end{align*}
where we set 
\begin{align*}
\mathcal{I}_{N,k}^{(1)}
=&
\cJ_{\rho}(\chi,\overline{\zeta}|\chi,\overline{\chi})\cJ_{\rho}(\eta,\overline{\chi}|\chi,\overline{\chi})
+
\cJ_{\rho}(\chi,\overline{\zeta}|\chi,\overline{\chi})\frac{\rho}{N}
\Bigl(\frac{k+1}{\sqrt{2\pi}}e^{-\frac{1}{2}(\eta+\overline{\chi}-\frac{\rho}{2})^2}+C(\eta,\overline{\chi})\Bigr)
\\
&
+
\cJ_{\rho}(\eta,\overline{\chi}|\chi,\overline{\chi})\frac{\rho}{N}
\Bigl(\frac{k+1}{\sqrt{2\pi}}e^{-\frac{1}{2}(\chi+\overline{\zeta}-\frac{\rho}{2})^2}+C(\chi,\overline{\zeta})\Bigr)
+
O(N^{-2}),
\\
\mathcal{I}_{N,k}^{(2)}
=&
\cJ_{\rho}(\eta,\overline{\zeta}|\chi,\overline{\chi})\cJ_{\rho}(\chi,\overline{\chi}|\chi,\overline{\chi})
+
\cJ_{\rho}(\eta,\overline{\zeta}|\chi,\overline{\chi})\frac{\rho}{N}
\Bigl(\frac{k+1}{\sqrt{2\pi}}e^{-\frac{1}{2}(\chi+\overline{\chi}-\frac{\rho}{2})^2}+C(\chi,\overline{\chi})\Bigr)
\\
&
+
\cJ_{\rho}(\chi,\overline{\chi}|\chi,\overline{\chi})\frac{\rho}{N}
\Bigl(\frac{k+1}{\sqrt{2\pi}}e^{-\frac{1}{2}(\eta+\overline{\zeta}-\frac{\rho}{2})^2}+C(\eta,\overline{\zeta})\Bigr)
+
O(N^{-2}).
\end{align*}
uniformly for $\zeta,\eta,\chi$ in compact subsets of $\C$. 
Since 
\begin{align*}
\mathcal{I}_{N,1}^{(1)}-\mathcal{I}_{N,0}^{(1)}
=&
\frac{\rho}{N}\cJ_{\rho}(\chi,\overline{\zeta}|\chi,\overline{\chi})\frac{e^{-\frac{1}{2}(\eta+\overline{\chi}-\frac{\rho}{2})^2}}{\sqrt{2\pi}}
+
\frac{\rho}{N}\cJ_{\rho}(\eta,\overline{\chi}|\chi,\overline{\chi})\frac{e^{-\frac{1}{2}(\chi+\overline{\zeta}-\frac{\rho}{2})^2}}{\sqrt{2\pi}}
+
O(N^{-2}),
\\
\mathcal{I}_{N,1}^{(2)}-\mathcal{I}_{N,0}^{(2)}
=&
\frac{\rho}{N}\cJ_{\rho}(\eta,\overline{\zeta}|\chi,\overline{\chi})\frac{e^{-\frac{1}{2}(\chi+\overline{\chi}-\frac{\rho}{2})^2}}{\sqrt{2\pi}}
+
\frac{\rho}{N}\cJ_{\rho}(\chi,\overline{\chi}|\chi,\overline{\chi})\frac{e^{-\frac{1}{2}(\eta+\overline{\zeta}-\frac{\rho}{2})^2}}{\sqrt{2\pi}}
+
O(N^{-2}),
\end{align*}
the first line in \eqref{FocusEq} becomes 
\begin{align*}
&\frac{N^2}{\rho^2}
\bigl(\mathrm{I}_{N+1}^{(b_N)}(\overline{z},w|\lambda,\overline{\lambda})
-\mathrm{I}_{N}^{(b_N)}(\overline{z},w|\lambda,\overline{\lambda})
\bigr)
\\
=&
\frac{N^2}{\rho^2}
e^{-\frac{1}{2}(|z|^2+|w|^2)+w\overline{z}}
\bigl(
e^{-(\overline{\zeta}-\overline{\chi})(\eta-\chi)}(\mathcal{I}_{N,1}^{(1)}-\mathcal{I}_{N,0}^{(1)})
-(1-(\overline{\zeta}-\overline{\chi})(\eta-\chi))(\mathcal{I}_{N,1}^{(2)}-\mathcal{I}_{N,0}^{(2)})
\bigr)
\nonumber
\\
=&
\frac{N}{\rho}\frac{e^{-\frac{1}{2}(|w|^2+|z|^2)+w\overline{z}}}{\sqrt{2\pi}}
\Bigl\{
e^{-(\overline{\zeta}-\overline{\chi})(\eta-\chi)}
\bigl(
\cJ_{\rho}(\chi,\overline{\zeta}|\chi,\overline{\chi})e^{-\frac{1}{2}(\eta+\overline{\chi}-\frac{\rho}{2})^2}
+
\cJ_{\rho}(\eta,\overline{\chi}|\chi,\overline{\chi})e^{-\frac{1}{2}(\chi+\overline{\zeta}-\frac{\rho}{2})^2}
\bigr)
\nonumber
\\
&
-(1-(\overline{\zeta}-\overline{\chi})(\eta-\chi))
\bigl(
\cJ_{\rho}(\eta,\overline{\zeta}|\chi,\overline{\chi})e^{-\frac{1}{2}(\chi+\overline{\chi}-\frac{\rho}{2})^2}
+
\cJ_{\rho}(\chi,\overline{\chi}|\chi,\overline{\chi})e^{-\frac{1}{2}(\eta+\overline{\zeta}-\frac{\rho}{2})^2}
\bigr)
\Bigr\}(1+O(N^{-1})),
\nonumber
\end{align*}
uniformly for $\zeta,\eta,\chi$ in compact subsets of $\C$. 
For the second line in the right hand side of \eqref{FocusEq}, we see that 
\begin{align*}
&\frac{N}{2}\mathrm{I}_{N+1}^{(b_N)}(\overline{z},w|\lambda,\overline{\lambda})
-
\frac{N}{\rho}(\chi+\overline{\chi})\mathrm{I}_{N}^{(b_N)}(\overline{z},w|\lambda,\overline{\lambda})
\\
=&
-\frac{N}{\rho}\bigl(\chi+\overline{\chi}-\frac{\rho}{2}\bigr)e^{-\frac{1}{2}(|z|^2+|w|^2)+w\overline{z}}
\Bigl(
e^{-(\overline{\zeta}-\overline{\chi})(\eta-\chi)}
\cJ_{\rho}(\chi,\overline{\zeta}|\chi,\overline{\chi})\cJ_{\rho}(\eta,\overline{\chi}|\chi,\overline{\chi})
\\
&\quad\quad\quad\quad\quad\quad\quad\quad\quad\quad\quad\quad
\quad\quad\quad
-(1-(\overline{\zeta}-\overline{\chi})(\eta-\chi))
\cJ_{\rho}(\eta,\overline{\zeta}|\chi,\overline{\chi})\cJ_{\rho}(\chi,\overline{\chi}|\chi,\overline{\chi})
\Bigr)(1+O(N^{-1})),
\end{align*}
uniformly for $\zeta,\eta,\chi$ in compact subsets of $\C$. 
Since there is no contribution from the third line in the right hand side of \eqref{FocusEq}, the contribution of \eqref{FocusEq} is 
\[
(N+b_N+1)\mathrm{I}_{N+1}^{(b_N)}(\overline{z},w|\lambda,\overline{\lambda})
-
\lambda\overline{\lambda}\mathrm{I}_{N}^{(b_N)}(\overline{z},w|\lambda,\overline{\lambda})
\\
=
\frac{N}{\rho}e^{-\frac{1}{2}(|z|^2+|w|^2)+w\overline{z}}
\mathcal{I}^{(3)}
(1+O(N^{-1})),
\]
where we set 
\begin{align*}
\mathcal{I}^{(3)}
=&
e^{-(\overline{\zeta}-\overline{\chi})(\eta-\chi)}
\bigl(
\cJ_{\rho}(\chi,\overline{\zeta}|\chi,\overline{\chi})\frac{e^{-\frac{1}{2}(\eta+\overline{\chi}-\frac{\rho}{2})^2}}{\sqrt{2\pi}}
+
\cJ_{\rho}(\eta,\overline{\chi}|\chi,\overline{\chi})\frac{e^{-\frac{1}{2}(\chi+\overline{\zeta}-\frac{\rho}{2})^2}}{\sqrt{2\pi}}
\bigr)
\\
&
-(1-(\overline{\zeta}-\overline{\chi})(\eta-\chi))
\bigl(
\cJ_{\rho}(\eta,\overline{\zeta}|\chi,\overline{\chi})\frac{e^{-\frac{1}{2}(\chi+\overline{\chi}-\frac{\rho}{2})^2}}{\sqrt{2\pi}}
+
\cJ_{\rho}(\chi,\overline{\chi}|\chi,\overline{\chi})\frac{e^{-\frac{1}{2}(\eta+\overline{\zeta}-\frac{\rho}{2})^2}}{\sqrt{2\pi}}
\bigr)
\\
&
-\bigl(\chi+\overline{\chi}-\frac{\rho}{2}\bigr)
\Bigl(
e^{-(\overline{\zeta}-\overline{\chi})(\eta-\chi)}
\cJ_{\rho}(\chi,\overline{\zeta}|\chi,\overline{\chi})\cJ_{\rho}(\eta,\overline{\chi}|\chi,\overline{\chi})
-(1-(\overline{\zeta}-\overline{\chi})(\eta-\chi))
\cJ_{\rho}(\eta,\overline{\zeta}|\chi,\overline{\chi})\cJ_{\rho}(\chi,\overline{\chi}|\chi,\overline{\chi})
\Bigr).
\end{align*}
Here, note that 
\[
\frac{\omega(z,\overline{z}|\lambda,\overline{\lambda})}{\varpi(\overline{z},w)}e^{-\frac{1}{2}(|z|^2+|w|^2)+\overline{z}w}
=
(1+(\overline{\zeta}-\overline{\chi})(\zeta-\chi))
e^{-\frac{1}{2}(\zeta+\overline{\zeta}+\frac{\rho}{2})^2+\frac{1}{2}(\eta+\overline{\zeta}+\frac{\rho}{2})^2}
(1+O(N^{-1})).
\]
Combining $\mathcal{I}^{(3)}$ with \eqref{KIA}, 
\begin{align*}
\mathfrak{I}_N^{(b_N)}(\overline{z},w|\lambda,\overline{\lambda})\omega(z,\overline{z}|\lambda,\overline{\lambda})
=&\big(\chi+\overline{\chi}+\frac{\rho}{2}\bigr)
\frac{(1+|\zeta-\chi|^2)
e^{-\frac{1}{2}(\zeta+\overline{\zeta}+\frac{\rho}{2})^2+\frac{1}{2}(\eta+\overline{\zeta}+\frac{\rho}{2})^2}}
{\cL_{\rho}(\chi+\overline{\chi})(\overline{\zeta}-\overline{\chi})^2(\eta-\chi)^2}
\mathcal{I}^{(3)}
(1+O(N^{-1}))
\\
=&
\frac{(1+|\zeta-\chi|^2)
e^{-\frac{1}{2}(\zeta+\overline{\zeta}+\frac{\rho}{2})^2+\frac{1}{2}(\eta+\overline{\zeta}+\frac{\rho}{2})^2}}
{\cL_{\rho}(\chi+\overline{\chi})(\overline{\zeta}-\overline{\chi})^2(\eta-\chi)^2}
\mathcal{I}^{(4)}
(1+O(N^{-1})),
\end{align*}
where we set
\begin{align*}
\mathcal{I}^{(4)}
=&
e^{-(\overline{\zeta}-\overline{\chi})(\eta-\chi)}
\cH_{\rho}(\chi+\overline{\chi},\chi+\overline{\zeta},\eta+\overline{\chi},\eta+\overline{\zeta},(\overline{\zeta}-\overline{\chi})(\eta-\chi))
\\
&
+(\overline{\zeta}-\overline{\chi})(\eta-\chi)
\Bigl\{
\bigl(\chi+\overline{\chi}+\frac{\rho}{2}\bigr)L_{\rho}(\chi+\overline{\chi})\frac{e^{-\frac{1}{2}(\eta+\overline{\zeta}-\frac{\rho}{2})^2}}{\sqrt{2\pi}}
-
\bigl(\chi+\overline{\chi}-\frac{\rho}{2}\bigr)L_{\rho}(\chi+\overline{\chi})\frac{e^{-\frac{1}{2}(\eta+\overline{\zeta}+\frac{\rho}{2})^2}}{\sqrt{2\pi}}
\\
&+\frac{e^{-\frac{1}{2}(\eta+\overline{\zeta}+\frac{\rho}{2})^2-\frac{1}{2}(\chi+\overline{\chi}-\frac{\rho}{2})^2}}{\sqrt{2\pi}^2}
+\frac{e^{-\frac{1}{2}(\eta+\overline{\zeta}-\frac{\rho}{2})^2-\frac{1}{2}(\chi+\overline{\chi}+\frac{\rho}{2})^2}}{\sqrt{2\pi}^2}
-\frac{\chi+\overline{\chi}-\frac{\rho}{2}}{\chi+\overline{\chi}+\frac{\rho}{2}}
\frac{e^{-\frac{1}{2}(\eta+\overline{\zeta}+\frac{\rho}{2})^2-\frac{1}{2}(\chi+\overline{\chi}+\frac{\rho}{2})^2}}{\sqrt{2\pi}^2}
\Bigr\}.
\end{align*} 
Note that the second line and third lines of the right hand side in the above are canceled by \eqref{WII} and \eqref{WIII}. 
Hence, we recall \eqref{cH}, and let us write $K_{1,1}^{(\mathrm{weak})}(\zeta,\overline{\zeta},\eta,\overline{\eta}|\overline{\chi},\chi)=\lim_{N\to\infty}K_{1,1}^{(N)}(z,\overline{z},w,\overline{w}|\lambda,\overline{\lambda})$, and then, we obtain 
\begin{align*}
K_{1,1}^{(\mathrm{weak})}(\zeta,\overline{\zeta},\eta,\overline{\eta}|\overline{\chi},\chi)
=&
\frac{\cH_{\rho}(\chi+\overline{\chi},\chi+\overline{\zeta},\eta+\overline{\chi},\eta+\overline{\zeta},(\overline{\zeta}-\overline{\chi})(\eta-\chi))}{\cL_{\rho}(\chi+\overline{\chi})(\overline{\zeta}-\overline{\chi})^2(\eta-\chi)^2}
e^{-(\overline{\zeta}-\overline{\chi})(\eta-\chi)-\frac{1}{2}(\zeta+\overline{\zeta}+\frac{\rho}{2})^2+\frac{1}{2}(\eta+\overline{\zeta}+\frac{\rho}{2})^2}
(1+|\zeta-\chi|^2)
\\
=&
\phi_{\rho}^{-1}(\zeta)
\frac{\cH_{\rho}(\chi+\overline{\chi},\chi+\overline{\zeta},\eta+\overline{\chi},\eta+\overline{\zeta},(\overline{\zeta}-\overline{\chi})(\eta-\chi))}{\cL_{\rho}(\chi+\overline{\chi})(\overline{\zeta}-\overline{\chi})^2(\eta-\chi)^2}
(1+|\zeta-\chi|^2)e^{-|\zeta-\chi|^2}
\phi_{\rho}(\eta),
\end{align*}
where $\phi_{\rho}(\eta)=e^{\overline{\chi}\eta+\frac{\rho}{2}\eta+\frac{\eta^2}{2}}$, which does not affect the value of the determinant. Here, the convergence is uniform for $\zeta,\eta,\chi$ in compact subsets of $\C$.
This completes the first part of the proof of the weakly non-unitary regime. 
Next, we consider the off-diagonal overlap case. 
By the decoupling Lemma~\ref{thm:lem1}, 
it suffices to compute $\Psi_{1,2}^{(\mathrm{weak})}(\zeta_1,\zeta_2)$. 
Using \eqref{AW1} and \eqref{AW2} and from the above discussions, we see that for $z_1=\sqrt{Na_N}+\zeta_1,z_2=\sqrt{Na_N}+\zeta_2$, as $N\to\infty$, 
\begin{align*}
&f_{N-1}^{(b_N)}(z_1\overline{z_2})
\varpi(z_1,\overline{z_2})
\varpi(\overline{z_1},z_2)
\cK^{(N-1)}(\overline{z_1},z_2|z_1,\overline{z_2})
\\
=&
e^{-|\zeta_1-\zeta_2|^2}
\frac{\cH_{\rho}(\zeta_1+\overline{\zeta_2},\zeta_1+\overline{\zeta_1},\zeta_2+\overline{\zeta_2},\zeta_2+\overline{\zeta_1},-(\overline{\zeta_1}-\overline{\zeta_2})(\zeta_1-\zeta_2))}
{(\overline{\zeta_1}-\overline{\zeta_2})^2(\zeta_1-\zeta_2)^2}(1+o(1)).
\end{align*}
This completes the proof of the second part of the weakly non-unitary regime. 
\subsection{Proof of the singular origin case in Theorem~\ref{MainThmS}}
We fix $a_N=1$ and $b_N=b>0$ as the parameters. 
Let $z=\zeta,w=\eta,\lambda=\chi$. 
Then, observe that 
\begin{equation}
\label{AS1}
\mathfrak{e}_{N+k-1}^{(b)}(z\wbar|\lambda\overline{\lambda})
\varpi(z,\overline{w})
=
\left((\zeta\overline{\eta})^bE_{1,b+1}(\zeta\overline{\eta})+\frac{1}{\Gamma(b)}\frac{(\zeta\overline{\eta})^b}{|\chi|^2-b}\right)
e^{-\frac{1}{2}(|\zeta|^2+|\eta|^2)},
\end{equation}
uniformly for $\zeta,\eta,\chi$ in compact subsets of $\C$.
Based on this asymptotic expansion, we shall observe the asymptotic behavior of each term again. 
First, note that 
\begin{equation}
\label{AS2}
f_N^{(b)}(\lambda\overline{\lambda})
\varpi(\lambda,\overline{\lambda})
=
\frac{N}{\chi\overline{\chi}}
\cE_b(\chi\overline{\chi}|\chi\overline{\chi})
(\chi\overline{\chi})^{b}e^{-|\chi|^2}
(1+O(N^{-1})),
\end{equation}
uniformly for $\zeta,\eta,\chi$ in a compact subset of $\C$ and $\cE_b$ is defined by \eqref{cEb}.
With these asymptotics, we firstly observe that as $N\to\infty$,
\begin{align*}
\mathrm{I}_N^{(b)}(\overline{z},w|\lambda,\overline{\lambda})
=
\Bigl(
\cE_b(\overline{\zeta}\chi|\chi\overline{\chi})\cE_b(\overline{\chi}\eta|\chi\overline{\chi})
-(1-(\overline{\zeta}-\overline{\chi})(\eta-\chi))
\cE_b(\overline{\zeta}\eta|\chi\overline{\chi})\cE_b(\overline{\chi}\chi|\chi\overline{\chi})
\Bigr)
\frac{(\overline{\zeta}\eta)^b(\chi\overline{\chi})^be^{-\frac{1}{2}(|\zeta|^2+|\eta|^2)-|\chi|^2}}{(|\chi|^2-b)^2},
\end{align*}
uniformly for $\zeta,\eta,\chi$ in a compact subset of $\C$.
Then, we have 
\begin{align*}
\mathfrak{I}_N^{(b)}(\overline{z},w|\lambda,\overline{\lambda})
=&
\frac{\cE_b(\overline{\zeta}\chi|\chi\overline{\chi})\cE_b(\overline{\chi}\eta|\chi\overline{\chi})
-(1-(\overline{\zeta}-\overline{\chi})(\eta-\chi))
\cE_b(\overline{\zeta}\eta|\chi\overline{\chi})\cE_b(\overline{\chi}\chi|\chi\overline{\chi})}
{(\chi\overline{\chi}-b)(\overline{\zeta}-\overline{\chi})^2(\eta-\chi)^2\cE_b(\chi\overline{\chi}|\chi\overline{\chi})}
(1+o(1)).
\end{align*}
Similarly, we have 
\[
\mathfrak{II}_N^{(b)}(\overline{z},w|\lambda,\overline{\lambda})
=
O(e^{-\epsilon N}),
\quad
\mathfrak{III}_N^{(b)}(\overline{z},w|\lambda,\overline{\lambda})
=
-\frac{1}{\Gamma(b)}\frac{1}{(\chi\overline{\chi}-b)(\overline{\zeta}-\overline{\chi})(\eta-\chi)}(1+o(1)).
\]
Hence, we obtain 
\begin{align*}
\cK^{(N)}(\overline{z},w|\lambda,\overline{\lambda})
=&
\frac{\cS_b(\overline{\zeta}\chi,\overline{\chi}\eta,\overline{\zeta}\eta,\overline{\chi}\chi,(\overline{\zeta}-\overline{\chi})(\eta-\chi))}
{(\overline{\zeta}-\overline{\chi})^2(\eta-\chi)^2\cE_b(\chi\overline{\chi})}(1+o(1)).
\end{align*}
uniformly for $\zeta,\eta,\chi$ in a compact subset of $\C$, and we recall \eqref{cS}.
Let 
\[
K_{11}^{(\mathrm{sing})}(\zeta,\overline{\zeta},\eta,\overline{\eta}|\chi,\overline{\chi})
=
\lim_{N\to\infty}
K_{11}^{(N)}(z,\overline{z},w,\overline{w}|\lambda,\overline{\lambda}).
\]
Here, from the discussion so far, it follows that the convergence is uniform for $\zeta,\eta,\chi$ in compact subsets of $\C$. 
Similarly, we consider the off-diagonal case. 
By the decoupling Lemma \ref{thm:lem1}, it suffices to consider 
\[
f_N^{(\alpha)}(z_1\overline{z_2})\varpi(z_1,\overline{z_2})\varpi(\overline{z_1},z_2)\cK^{(N-1)}(\overline{z_1},z_2|z_1,\overline{z_2}).
\]
However, it is straightforward to see that 
\[
f_N^{(b)}(z_1\overline{z_2})
=
\frac{N}{\zeta_1\overline{\zeta_2}}\cE_b(\zeta_1\overline{\zeta_2})(1+O(N^{-1})),
\quad
\varpi(z_1,\overline{z_2})\varpi(\overline{z_1},z_2)
=
|\zeta_1|^{2b}|\zeta_2|^{2b}e^{-(|\zeta_1|^2+|\zeta_2|^2)},
\]
and 
\begin{align*}
\cK^{(N-1)}(\overline{z_1},z_2|z_1,\overline{z_2})
=&
\cK^{(\mathrm{sing})}(\overline{\zeta_1},\zeta_2|\zeta_1,\overline{\zeta_2})(1+o(1))
\\
=&
\frac{\cS_b(\overline{\zeta_1}\zeta_1,\overline{\zeta_2}\zeta_2,\overline{\zeta_1}\zeta_2,\overline{\zeta_2}\zeta_1,(\overline{\zeta_1}-\overline{\zeta_2})(\zeta_2-\zeta_1))}
{(\overline{\zeta_1}-\overline{\zeta_2})^2(\zeta_2-\zeta_1)^2\cE_b(\zeta_1\overline{\zeta_2})}(1+o(1)).
\end{align*}
Hence, we have that as $N\to\infty$,
\begin{align*}
f_N^{(b)}(z_1\overline{z_2})\varpi(z_1,\overline{z_2})\varpi(\overline{z_1},z_2)\cK^{(N-1)}(\overline{z_1},z_2|z_1,\overline{z_2})
\to&
\frac{\cS_b(\overline{\zeta_1}\zeta_1,\overline{\zeta_2}\zeta_2,\overline{\zeta_1}\zeta_2,\overline{\zeta_2}\zeta_1,(\overline{\zeta_1}-\overline{\zeta_2})(\zeta_2-\zeta_1))}{\zeta_1\overline{\zeta_2}(\overline{\zeta_1}-\overline{\zeta_2})^2(\zeta_2-\zeta_1)^2}|\zeta_1|^{2b}|\zeta_2|^{2b}e^{-(|\zeta_1|^2+|\zeta_2|^2)}
\\
=&
\Psi_{1,2}^{(\mathrm{sing})}(\zeta_1,\zeta_2). 
\end{align*}
This completes the proof of singular origin case in Theorem~\ref{MainThmS}.
\section{Concluding remarks}
In this paper, we studied the multi-point intensity of the on- and off-diagonal overlap of the induced Ginibre unitary ensemble. 
As a consequence, we obtained the new scaling limits of the local statistics associated with the overlap weight function and the mean of on- and off-diagonal overlaps. 
Its deviation is based on the moment method since the planar orthogonal polynomials associated with the weight function can not be constructed from the method in \cite{AV}. 
The interesting point is that our planar orthogonal polynomials satisfy the non-standard three-term recurrence relationship pointed out in \cite[Remark 1.3]{SungsooLee},
and it would be interesting to make the similarity of such structure clear and generalize the results in this paper into the non-Hermitian random matrix with non-radially symmetric potential such as the elliptic Ginibre unitary ensemble. 
Also, we only studied the on- and off-diagonal overlap, but it is a natural future direction to study the correlations of the overlaps such as $\E[\cO_{1,1}\cO_{2,2}\cO_{3,3}]$ and $\E[\cO_{1,1}\cO_{1,2}\cO_{2,3}]$ conditionally on the eigenvalues $z_1=\zeta,z_2=\eta,z_3=\lambda$. Finally, by replacing the standard complex Gaussian random variables with the complex Brownian motion entries in the construction of the induced Ginibre unitary ensemble, it would be interesting to study the stochastic dynamics of the eigenvalues and the overlaps for the time-evolutional induced Ginibre unitary ensemble as in \cite{Esaki,Yabuoku}.

\section*{Acknowledgements}
I gratefully acknowledge the continuous encouragement and supports from my supervisor Professor Tomoyuki Shirai. 
I am deeply grateful to Professor Gernot Akemann, Professor Sung-Soo Byun, and Satoshi Yabuoku for useful discussions. 
A part of this work was done during a visit to Bielefeld University and Zentrum f\"{u}r interdisziplin\"{a}re Forschung (ZiF). 
I gratefully acknowledge my host Professor Gernot Akemann during my stay in Bielefeld University, and 
I am deeply grateful to both institutions for their warm hospitality. 
This work was supported by JSPS KAKENHI Grant Number (B) 18H01124 and 23H01077 and WISE program (JSPS) at Kyushu University and the Deutsche Forschungsgemeinschaft (DFG) grant SFB 1283/2 2021--317210226.

\bibliographystyle{amsplain}

\begin{thebibliography}{10}

 \bibitem{ABC}
G. Akemann, M. Baake, N. Chakarov, O. Krger, A. Mielke, M. Ottensmann and R. Werdehausen.: 
Territorial behaviour of buzzards versus
random matrix spacing distributions, J. Theor. Biol. 509 (2021), 110475.

 \bibitem{ATTZ}
G. Akemann, R. Tribe, A. Tsareas, and O. Zaboronski, On the determinantal structure of conditional
overlaps for the complex Ginibre ensemble, Random Matrices: Theory and Applications 9 (2020), no. 04,
2050015. MR4133071

 \bibitem{Akemann_Foster_Kieburg:2020} 
G. Akemann, Y. F\"{o}rster and M. Kieburg.: 
Universal eigenvector correlations in quaternionic Ginibre ensembles, J. Phys. A., 53, (2020), 145201.

\bibitem{Akemann_2014}
G. Akemann and M. J. Phillips. The interpolating Airy kernels for the $\beta=1$ and $\beta=4$ elliptic Ginibre ensembles.
J. Stat. Phys., 155(3):421--465, 2014.

\bibitem{AV}
G. Akemann and G. Vernizzi.: Characteristic polynomials of complex random matrix models.
In: Nuclear Physics B 660.3 (2003), pp. 532--556.

\bibitem{Ameur_2021}
Y. Ameur and S.-S. Byun.: 
Almost-Hermitian random matrices and bandlimited point processes. 
Anal. Math. Phys.13 (2023), 52.


\bibitem{BFD}
R. Bardenet, A. Feller, J. Bouttier, P. Degiovanni, A. Hardy, A. Rancon, B. Roussel, G. Schehr, and C. I. Westbrook. 
From point processes to quantum optics and back. arXiv:2210.05522, 2022.

 \bibitem{Belinschi_Nowak_Speicher_Tarnowski:200}
S. Belinschi, M.A. Nowak, R. Speicher and W. Tarnowski.:
Squared eigenvalue condition numbers and eigenvector correlations from the single
ring theorem, J. Phys. A 50 (2017), 105204.


\bibitem{Borodin:1998}
A. Borodin.: Biorthogonal ensembles, Nucl. Phys. B 536 (1998), no. 3, 704--732


\bibitem{Bourgade_Dubach:2021}  
P. Bourgade and G. Dubach.:
The distribution of overlaps between eigenvectors of Ginibre matrices, Probab. Theory Relat. Fields 177 (2020),
397--464.


\bibitem{BZ}
F. Benaych-Georges and O. Zeitouni.:
Eigenvectors of non normal random matrices. Electron. Commun. Probab. (2018) {\it 23}, 1--12.


\bibitem{BSV:2017}
Z. Burda, B.J. Spisak and P. Vivo.: Eigenvector statistics of the product of Ginibre matrices, Phys. Rev. E, 95 (2017), 022134.

\bibitem{Sungsoo1_2022}
S.-S. Byun and C. Charlier. On the almost-circular symplectic induced Ginibre ensemble. Stud. Appl. Math. (Online),
arXiv:2206.06021, 2022.

\bibitem{Sungsoo2_2022}
S.-S. Byun and P. J. Forrester. Progress on the study of the Ginibre ensembles I: GinUE. {\it preprint arXiv:2211.16223}, 2022.

\bibitem{Sungsoo3_2022}
S.-S. Byun and P. J. Forrester. Spherical induced ensembles with symplectic symmetry. {\it preprint arXiv:2209.01934}, 2022.

\bibitem{Sungsoo4_2022}
S.-S. Byun and P.J. Forrester, Progress on the study of the Ginibre ensembles II: GinOE and GinSE,
{\it arXiv:2301.05022}, 2023.

\bibitem{SungsooLee}
S.-S. Byun, S.-Y. Lee and M. Yang, Lemniscate ensembles with spectral singularity, arXiv:2107.07221.

\bibitem{Sungsoo5_2022}
S.-S. Byun, N.-G. Kang, and S.-M. Seo. :
Partition Functions of Determinantal and Pfaffian Coulomb Gases with Radially Symmetric Potentials. Commun. Math. Phys. (2023). https://doi.org/10.1007/s00220-023-04673-1

\bibitem{Sungsoo6_2022}
S.-S. Byun and S.-M. Seo.:
Random normal matrices in the almost-circular regime. Bernoulli 29 (2) 1615--1637, May 2023. https://doi.org/10.3150/22-BEJ1514


\bibitem{Cipolloni22}
C. Cipolloni and D. Schr\"{o}der.:
On the condition number of the shifted real Ginibre ensemble. 
SIAM J. Matr. Anal. Applic. (2022) {\it 43}, 1469--1487.

\bibitem{Cipolloni23c}
C. Cipolloni, L. Erd\"{o}s, J. Henheik, and D. Schr\"{o}der.:
Optimal Lower Bound on Eigenvector Overlaps for non-Hermitian Random Matrices, 2023, arXiv:2301.03549


\bibitem{CM_1998}
J. T. Chalker and B. Mehlig.: Eigenvector statistics in non-Hermitian random matrix ensembles, Phys. Rev. Lett. 81 (1998), 3367--3370.

\bibitem{CM_2000}
J. T. Chalker and B. Mehlig.: Statistical properties of eigenvectors in non-Hermitian Gaussian random matrix ensembles, J. Math. Phys. 41 (2000),
3233--3256.

\bibitem{CR2022}
N;. Crawford and R. Rosenthal.: 
Eigenvector correlations in the complex Ginibre ensemble. Ann. Appl. Probab.
(2022) 32, 2706--2754.



 \bibitem{Dubach_2021w1}
G. Dubach, On eigenvector statistics in the spherical and truncated unitary ensembles, Elec. J. Probab. 26 (2021), 1--29.

\bibitem{Dubach_2021w2}
G. Dubach.: 
Symmetries of the quaternionic Ginibre ensemble, Random Matrices Theory Appl., 10 (2021), 2150013.

\bibitem{Dubach_2023}
G. Dubach.: Explicit formulas concerning eigenvectors of weakly non-unitary matrices. Electron. Commun.
Probab. (2023) 28 1--11.

\bibitem{ErdJi} 
L. Erd\"{o}s and H.C. Ji.:
Wegner estimate and upper bound on the eigenvalue condition number of non-Hermitian random matrices. arXiv:2301.04981. (2023)

\bibitem{Esaki}
S. Esaki, M. Katori, and S. Yabuoku.:
Eigenvalues, eigenvector-overlaps, and regularized Fuglede-Kadison
determinant of the non-Hermitian matrix-valued Brownian motion. arXiv 2023, arXiv:2306.00300.


\bibitem{Fischmann}
J. Fischmann, W. Bruzda, B. A. Khoruzhenko, H.-J. Sommers, and K. Zyczkowski.
Induced Ginibre ensemble of random matrices and quantum operations. J. Phys. A, 45(7):075203, 31, 2012.

\bibitem{Forrester_2010}
P.J. Forrester.: Log-gases and random matrices, Princeton University Press, Princeton, NJ, 2010.

\bibitem{Fyodorov_1997}
Y. V. Fyodorov, B. A. Khoruzhenko, and H.-J. Sommers.: 
Almost Hermitian random matrices: crossover from Wigner-Dyson to Ginibre eigenvalue statistics. 
Phys. Rev. Lett., 79(4):557--560, 1997.

\bibitem{Fyodorov_1998}
Y. V. Fyodorov, B. A. Khoruzhenko and H.-J. Sommers.: 
Universality in the random matrix spectra in the regime of weak non-Hermiticity.
Ann. Inst. H. Poincar\'{e} Phys. Th\'{e}or., 68 (1998), 449--489.

\bibitem{FM}
Y.V. Fyodorov and B. Mehlig.:
Statistics of resonances and nonorthogonal eigenfunctions in a model for single-channel chaotic scattering. 
Phys. Rev. E. (2002) {\it 66}, 045202.

\bibitem{FSav2}
Y.V. Fyodorov and D.V. Savin.:
Statistics of resonance width shifts as a signature of eigenfunction non-orthogonality.
Phys. Rev. Lett. (2012) {\it 108}, 184101.

\bibitem{Fyodorov_2018}
Y. V. Fyodorov.: On statistics of bi-orthogonal eigenvectors in real and complex Ginibre ensembles: combining partial Schur decomposition with
supersymmetry, Commun. Math. Phys. 363 (2018), 579--603.

\bibitem{Fyodorov_2021}
Y.V. Fyodorov and W. Tarnowski.: 
Condition numbers for real eigenvalues in the real elliptic Gaussian ensemble, Ann. Henri Poincar\'{e}, 22, (2021),
309--330.


\bibitem{Fyodorov_2003}
Y.V. Fyodorov and H.-J. Sommers.: Random matrices close to hermitian or unitary: overview of methods and results, J. Phys. A 36 (2003), 3303--3347


\bibitem{Ginibre_1965}
J. Ginibre.: Statistical ensembles of complex, quaternion, and real matrices. J. Math. Phys. 6, 440--449 (1965)

\bibitem{Grela_2018}
J. Grela and P. Warchol.: Full Dysonian dynamics of the complex Ginibre ensemble, J. Phys. A 51 (2018), 42.


\bibitem{HM2013}
H. Hedenmalm, N. Makarov.: 
Coulomb gas ensembles and Laplacian growth, Proc. London Math Soc. 106 (2013), 859--907.

\bibitem{Hough_2006}
J.B. Hough, M. Krishnapur, Y. Peres and B. Vir\'{a}g.: 
Zeros of Gaussian analytic functions and determinantal point processes, American Mathematical Society, Providence, RI, 2009.


\bibitem{JNNPZ}
R.A. Janik, W. N\"{o}renberg, M.A. Nowak, G. Papp, and I. Zahed.:
Correlations of eigenvectors for nonHermitian random matrix models.
Phys. Rev. E. (1999) {\it 60}, 2699--2705

\bibitem{KT}
A. Kulesza, B. Taskar, et al. Determinantal point processes for machine learning. Foundations and Trends in Machine Learning, 5(2--3):123, 2012.

\bibitem{Mehta_1991}
M.L. Mehta.: Random matrices, 2nd ed., Academic Press, New York, 1991.

\bibitem{MS}
B. Mehlig and M. Santer.:
Universal eigenvector statistics in a quantum scattering ensemble.
Phys. Rev. E. (2001) {\it 63}, 020105(R).

\bibitem{noda}
K. Noda.: Determinantal structure of the overlaps of induced spherical unitary ensemble, in preparation. 

\bibitem{Nowak_2018}
M. A. Nowak and W. Tarnowski.: Probing non-orthogonality of eigenvectors in non-Hermitian matrix models: diagrammatic approach, JHEP (2018), 152.


\bibitem{OLB}
F. W. Olver, D. W. Lozier, R. F. Boisvert, and C. W. Clark (Editors). 
NIST Handbook of Mathematical Functions.
Cambridge University Press, Cambridge, 2010.

\bibitem{RVW}
S. O'Rourke, V. Vu, and K. Wang.: 
Eigenvectors of random matrices: a survey. 
J. Combin. Theory.
Ser. A 144 (2016), 361--442.

\bibitem{SFT}
E. B. Saff and V. Totik.: 
{\it Logarithmic potentials with external fields}, 
Springer-Verlag, Berlin, 1997.

\bibitem{ST1}
T. Shirai and Y. Takahashi.: 
Random point fields associated with certain Fredholm determinants. I. fermion, Poisson and boson point processes. 
J. Funct. Anal., 205 (2003) 414--463.

\bibitem{Walter_Starr_2015}
M. Walters and S. Starr.: A note on mixed matrix moments for the complex Ginibre ensemble, J. Math. Phys. 56 (2015), 013301.

\bibitem{WW}
C. Webb and M.D. Wong.: 
On the moments of the characteristic polynomial of a Ginibre random matrix.
Proc. Lond. Math. Soc. 118 (2019),
1017--1056.

\bibitem{WTF}
T.R. W\"{u}rfel, M.J. Crumpton, and Y.V. Fyodorov.:
Mean left-right eigenvector self-overlap in the real Ginibre ensemble.
arXiv 2023, arXiv:2310.04307.

\bibitem{Yabuoku}
S. Yabuoku, Eigenvalue processes of Elliptic Ginibre Ensemble and their Overlaps, Int. J.
Math. Ind. 12, 1, 2050003 (2020)
\end{thebibliography}

\end{document}